\documentclass[a4paper,UKenglish,hidelinks]{lipics-v2018}
\usepackage{graphicx}
\usepackage[usenames,dvipsnames]{xcolor}
\usepackage{subcaption}
\usepackage{amsmath}
\usepackage{amsthm}
\usepackage{url,xspace}
\usepackage{enumitem}
\usepackage{verbatim}
\usepackage{tabularx,booktabs}

\nolinenumbers

\newtheorem{observation}[theorem]{\textbf{Observation}}

\graphicspath{{figures/}}

\newcommand{\mkmcal}[1]{\ensuremath{\mathcal{#1}}\xspace}

\newcommand{\A}{\mkmcal{A}}
\newcommand{\RR}{\mkmcal{R}}

\newcommand{\+}{\diagup}
\renewcommand{\-}{\diagdown}

\newcommand{\etal}{et al.\xspace}

\newcommand{\slope}{\ensuremath{\mathit{slope}}\xspace}

\newcommand{\mkmbb}[1]{\ensuremath{\mathbb{#1}}\xspace}
\newcommand{\R}{\mkmbb{R}}

\newcommand{\Both}{\ensuremath{\mathord\Updownarrow}}
\newcommand{\Upper}{\ensuremath{\mathord\Uparrow}}
\newcommand{\Lower}{\ensuremath{\mathord\Downarrow}}

\title{Convex partial transversals of planar regions}

\author{Vahideh Keikha}{Department of Mathematics and Computer Science, University of Sistan and Baluchestan, Zahedan, Iran}{va.keikha@gmail.com}{}{}

\author{Mees van de Kerkhof}{Department of Information and Computing Sciences, Utrecht
  University, \\{Utrecht, The Netherlands}}{m.a.vandekerkhof@uu.nl}{}{M.v.d.K.
  supported by the Netherlands Organisation for Scientific Research under
  proj. 628.011.005.}
\author{Marc van Kreveld}{Department of Information and Computing Sciences, Utrecht
  University, \\{Utrecht, The Netherlands}}{m.j.vankreveld@uu.nl}{}{M.v.K.
  supported by the Netherlands Organisation for Scientific Research under
  proj. 612.001.651.}
\author{Irina Kostitsyna}{Department of Mathematics and Computer Science, TU Eindhoven, \\{Eindhoven, The Netherlands}}{i.kostistyna@tue.nl}{}{}

\author{Maarten Löffler}{Department of Information and Computing Sciences, Utrecht
  University, \\{Utrecht, The Netherlands}}{m.loffler@uu.nl}{}{M.L.
  supported by the Netherlands Organisation for Scientific Research under
  proj. 614.001.504.}
\author{Frank Staals}{Department of Information and Computing Sciences, Utrecht
  University, \\{Utrecht, The Netherlands}}{f.staals@uu.nl}{}{F.S.
  supported by the Netherlands Organisation for Scientific Research under
  proj. 612.001.651.}
\author{Jérôme Urhausen}{Department of Information and Computing Sciences, Utrecht
  University, \\{Utrecht, The Netherlands}}{j.e.urhausen@uu.nl}{}{J.U.
  supported by the Netherlands Organisation for Scientific Research under
  proj. 612.001.651.}
\author{Jordi L. Vermeulen}{Department of Information and Computing Sciences, Utrecht
  University, \\{Utrecht, The Netherlands}}{j.l.vermeulen@uu.nl}{}{J.V.
  supported by the Netherlands Organisation for Scientific Research under
  proj. 612.001.651.}
\author{Lionov Wiratma}{Department of Information and Computing Sciences, Utrecht
  University\\{Utrecht, The Netherlands}\\Department of Informatics, Parahyangan
  Catholic University\\{Bandung, Indonesia}}{l.wiratma@uu.nl;lionov@unpar.ac.id}{}
  {L.W.  supported by the Mnst. of Research, Techn. and High. Ed.
  of Indonesia (No. 138.41/E4.4/2015)}

\authorrunning{Keikha, Kerkhof, Kreveld, Kostitsyna, L\"offler, Staals, Urhausen, Vermeulen, Wiratma}

\Copyright{Vahideh Keikha, Mees van der Kerkhof, Marc van Kreveld, Irina
  Kostitsyna, Maarten L\"offler, Frank Staals, Jérôme Urhausen, Jordi
  L. Vermeulen, Lionov Wiratma}

\subjclass{
	Theory of computation $\rightarrow$ Computational Geometry}

\hideLIPIcs  

\EventEditors{Wen-Lian Hsu, Der-Tsai Lee, and Chung-Shou Liao}
\EventNoEds{3}
\EventLongTitle{29th International Symposium on Algorithms and Computation (ISAAC 2018)}
\EventShortTitle{ISAAC 2018}
\EventAcronym{ISAAC}
\EventYear{2018}
\EventDate{December 16--19, 2018}
\EventLocation{Jiaoxi, Yilan, Taiwan}
\EventLogo{}
\SeriesVolume{123}
\ArticleNo{187}

\begin{document}

\maketitle

\begin{abstract}
  We consider the problem of testing, for a given set of planar regions
  $\cal R$ and an integer $k$, whether there exists a convex shape whose
  boundary intersects at least $k$ regions of $\cal R$.  We provide a polynomial
  time algorithm for the case where the regions are disjoint line segments with a constant number of orientations.
  On the other hand, we show that the problem is NP-hard when the regions are
  intersecting axis-aligned rectangles or 3-oriented line segments.  For
  several natural intermediate classes of shapes (arbitrary disjoint segments,
  intersecting 2-oriented segments) the problem remains open.
\end{abstract}

\section{Introduction}

A set of points $Q$ in the plane is said to be in \emph{convex position} if for
every point $q \in Q$ there is a halfplane containing $Q$ that has $q$ on its
boundary. Now, let ${\cal R}$ 
be a set of $n$ {\em regions} in the plane.  We say that
$Q$ is a {\em partial transversal} of $\cal
R$ if there exists an injective map $f: Q \to \cal R$ such that $q \in
f(q)$ for all $q \in Q$; if $f$ is a bijection we call
$Q$ a {\em full transversal}.
In this paper, we are concerned with the question whether a given set of regions $\cal R$ admits a convex partial transversal $Q$ of a given cardinality $|Q| = k$.
Figure~\ref {fig:example} shows an example.

\begin{figure}[tb]
  \centering
  \includegraphics{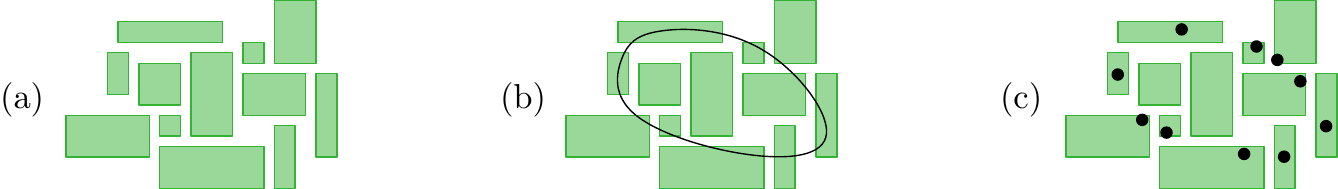}
  \caption{(a) A set of 12 regions. (b, c) A convex partial transversal of size $10$.}
  \label{fig:example}
\end{figure}


The study of convex transversals was initiated by Arik Tamir at the Fourth NYU Computational Geometry Day in 1987, who asked
``Given a collection of compact sets, can one decide in polynomial time whether there exists a convex body whose boundary intersects every set in the collection?''
Note that this is equivalent to the question of whether a convex full transversal of the sets exists: given the convex body, we can place a point of its boundary in every intersected region; 
conversely, the convex hull of a convex transversal forms a convex body whose boundary intersects every set.
In 2010, Arkin~\etal~\cite {ARKIN2014224} answered Tamir's original question in the negative (assuming P $\ne$ NP): they prove that the problem is NP-hard, even when the regions are (potentially intersecting) line segments in the plane, regular polygons in the plane, or balls in $\R^3$.
On the other hand, they show that Tamir's problem can be solved in polynomial time when the regions are {\em disjoint} segments in the plane and the convex body is restricted to be a polygon whose vertices are chosen from a given discrete set of (polynomially many) candidate locations. 
Goodrich and Snoeyink~\cite{gsch} show that for a set of {\em parallel} line segments, the existence of a convex transversal can be tested in $O(n \log n)$ time.
Schlipf~\cite {schlipf2012notes} further proves that the problem of finding a convex stabber for a set of disjoint {\em bends} (that is, shapes consisting of two segments joined at one endpoint) is also NP-hard.
She also studies the optimisation version of maximising the number of regions stabbed by a convex shape; we may re-interpret this question as finding the largest $k$ such that a convex partial transversal of cardinality $k$ exists.
She shows that this problem is also NP-hard for a set of (potentially intersecting) line segments in the plane.

\subparagraph{Related work.} Computing a partial transversal of maximum size
arises in wire layout applications~\cite{tompa1980optimal}. When each region in
$\cal R$ is a single point, our problem reduces to determining whether a point
set $P$ has a subset of cardinality $k$ in convex position.  Eppstein
\etal~\cite{50} solve this in $O(kn^3)$ time and $O(kn^2)$ space using dynamic
programming; the total number of convex $k$-gons can also be tabulated in
$O(kn^3)$ time~\cite{rote1991counting,mitchellcounting}.

If we allow reusing elements, our problem becomes equivalent to so-called \emph{covering color classes} introduced by Arkin \etal~\cite{arkin2015}.
Arkin \etal show that for a set of regions $\cal R$ where each region is a set of two or three points, computing a convex partial transversal of $\cal R$ of maximum cardinality is NP-hard.
Conflict-free coloring has been studied extensively, and has applications in, for instance, cellular networks~\cite{even2003conflict,har2005conflict,katz2012conflict}.

\begin{table}
\centering
\caption{New and known results.}
\label{table:results}
{\footnotesize
\begin{tabular}{ r rcc }
  \toprule
                        & & \bf disjoint & \bf intersecting \\
  \midrule
  \bf line segments:
  & parallel              & $O(n^{6})$ (upper hull only: $O(n^2)$) & N/A \\
  & 2-oriented            & \(\downarrow\) & open \\
  & 3-oriented            & \(\downarrow\) & NP-hard \\
  & $\rho$-oriented       & polynomial & $\uparrow$ \\
  & arbitrary             & open & NP-hard~\cite {ARKIN2014224} \\
  \midrule
  \bf rectangles:
  & squares               & open & open \\
  & rectangles            & open & NP-hard \\
  \midrule
  \bf other:
  & bends                 & NP-hard~\cite {schlipf2012notes} & $\leftarrow$ \\
  \bottomrule
\end{tabular}
}
\end{table}


\subparagraph{Results.}
Despite the large body of work on convex transversals and natural extensions of partial transversals that are often mentioned in the literature, surprisingly, no positive results were known.
We present the first positive results: in Section~\ref {sec:parallel} we show how to test whether a set of parallel line segments admits a convex transversal of size $k$ in polynomial time; we extend this result to disjoint segments of a fixed number of orientations in Section~\ref {sec:2-oriented}.
Although the hardness proofs of Arkin~\etal and Schlipf do extend to partial convex transversals, we strengthen these results by showing that the problem is already hard when the regions are $3$-oriented segments or axis-aligned rectangles (Section~\ref {sec:hardthreeintersect}).
Our results are summarized in Table~\ref {table:results}. The arrows in the table indicate that one result is implied by another.

For ease of terminology, in the remainder of this paper, we will drop the
qualifier ``partial'' and simply use ``convex transversal'' to mean ``partial
convex transversal''.  Also, for ease of argument, in all our results we test
for {\em weakly convex} transversals. This means that the transversal may
contain three or more colinear points.

\section{Parallel disjoint line segments}
\label{sec:parallel}

Let \RR be a set of $n$ vertical line segments in $\R^2$. We assume that no
three endpoints are aligned. Let $\Upper\RR$ and $\Lower\RR$ denote the sets of
upper and lower endpoints of the regions in \RR, respectively, and let
$\Both\RR=\Upper\RR\cup\Lower\RR$. In
Section~\ref{sub:Computing_an_upper_convex_transversal} we focus on computing
an \emph{upper convex transversal} --a convex transversal $Q$ in which all
points appear on the upper hull of $Q$-- that maximizes the number of regions
visited. We show that there is an optimal transversal whose strictly convex
vertices lie only on bottom endpoints in $\Lower\RR$. In
Section~\ref{sub:Computing_a_convex_transversal} we prove that there exists an
optimal convex transversal whose strictly convex vertices are taken from the
set of all endpoints $\Both\RR$, and whose leftmost and rightmost vertices are
taken from a discrete set of points. This leads to an $O(n^6)$ time dynamic
programming algorithm to compute such a transversal.

\subsection{Computing an upper convex transversal}
\label{sub:Computing_an_upper_convex_transversal}

Let $k^*$ be the maximum number of regions visitable by a upper convex
transversal of \RR.

\begin{lemma}
  \label{lem:discrete_upper_hull}
  Let $U$ be an upper convex transversal of \RR that visits $k$ regions. There
  exists an upper convex transversal $U'$ of \RR, that visits the same $k$
  regions as $U$, and such that the leftmost vertex, the rightmost vertex, and
  all strictly convex vertices of $U'$ lie on the bottom endpoints of the
  regions in \RR.
\end{lemma}

\begin{proof}
  Let $\mathcal{U}$ be the set of all upper convex transversals with $k$
  vertices. Let $U'\in\mathcal{U}$ be a upper convex transversal such that the
  sum of the $y$-coordinates of its vertices is minimal. Assume, by
  contradiction, that $U'$ has a vertex $v$ that is neither on the lower
  endpoint of its respective segment nor aligned with its adjacent
  vertices. Then we can move $v$ down without making the upper hull
  non-convex. This is a contradiction. Therefore, all vertices in $U'$ are
  either aligned with their neighbors (and thus not strictly convex), or at the
  bottom endpoint of a region.
\end{proof}

Let $\Lambda(v,w)$ denote the set of bottom endpoints of regions in \RR that
lie left of $v$ and below the line through $v$ and $w$. See
Figure~\ref{fig:upper_dp}(a). Let $\slope(\overline{uv})$ denote the slope of the
supporting line of $\overline{uv}$, and observe that
$\slope(\overline{uv})=\slope(\overline{vu})$.

By Lemma~\ref{lem:discrete_upper_hull} there is an optimal upper convex
transversal of \RR in which all strictly convex vertices lie on bottom endpoints
of the segments. Let $K[v,w]$ be the maximum number of regions visitable
by a
upper convex transversal that ends at a bottom endpoint $v$, and has an incoming
slope at $v$ of \emph{at least} $\slope(\overline{vw})$. It is important to
note that the second argument $w$ is used only to specify the slope, and $w$
may be left or right of $v$. We have that
\[ K[v,w] = \max_{u \in \Lambda(v,w)} \max_{s \in \Lambda(u,v)} K[u,s] + I[u,v], \]
where $I[u,v]$ denotes the number of regions in \RR intersected by the segment
$\overline{uv}$ (in which we treat the endpoint at $u$ as open, and the
endpoint at $v$ as closed). See Figure~\ref{fig:upper_dp}(a) for an illustration.

\begin{figure}[tb]
  \centering
  \includegraphics{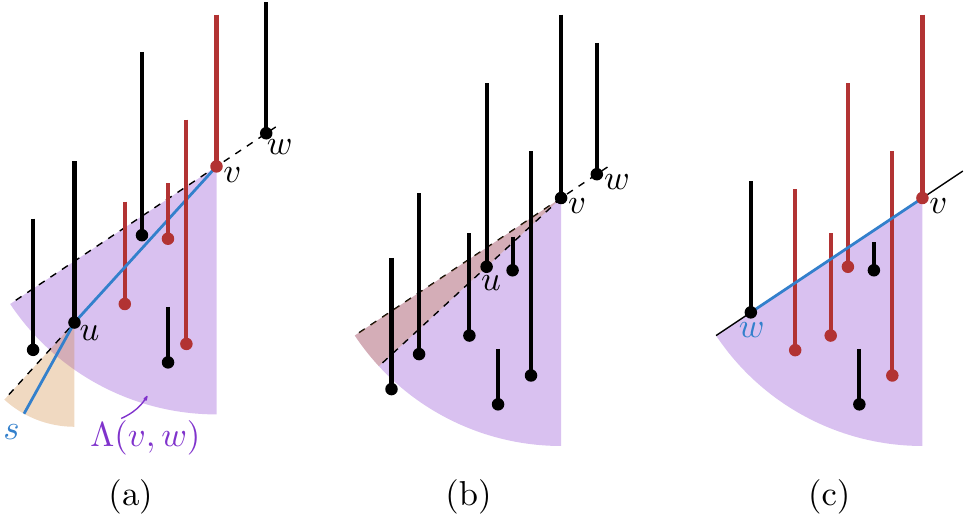}
  \caption{(a) The definition of $K[v,w]$. The region $\Lambda(v,w)$ is
    indicated in purple. The segments counted in $I[u,v]$ are shown in red. (b)
    The case that $K[v,w]=K[v,u]$, where $u$ corresponds to the predecessor
    slope of $\slope(\overline{vw})$. (c) The case that
    $K[v,w]=K[w,v]+I[w,v]$. }
  \label{fig:upper_dp}
\end{figure}

\newcommand{\pred}{\ensuremath{\mathit{pred}}\xspace}

\begin{observation}
  \label{obs:monotone_slopes}
  Let $v$, $s$, and $t$ be bottom endpoints of segments in \RR with
  $\slope(\overline{sv}) > \slope(\overline{tv})$. We have that
  $K[v,t] \geq K[v,s]$.
\end{observation}

Fix a bottom endpoint $v$, and order the other bottom endpoints
$w \in \Lower\RR$ in decreasing order of slope $\slope(\overline{wv})$. Let
$S_v$ denote the resulting order. We denote the $x$-coordinate of a point $v$ by $v_x$.

\begin{lemma}
  \label{lem:recurrence_upper}
  Let $v$ and $w$ be bottom endpoints of regions in \RR, and let $u$ be the
  predecessor of $w$ in $S_v$, if it exists (otherwise let
  $K[v,u]=-\infty$). We have that
  \begin{align*}
    K[v,w] = \begin{cases}
                \max\{1, K[v, u], K[w,v] + I[w,v]\} & \text{if }w_x < v_x\\
                \max\{1, K[v, u]\} & \text{otherwise.}
             \end{cases}
  \end{align*}
\end{lemma}

\begin{proof}
	If $w$ does not have any predecessor in $S_v$ then $w$ can be the only
	endpoint in $\Lambda(v,w)$. In particular, if $w$ lies right of $v$ then
	$\Lambda(v,w)$ is empty, and thus $K[v,w]=1$, i.e. our transversal starts and
	ends at $v$. If $w$ lies left of $v$ we can either visit only $v$ or arrive
	from $w$, provided the incoming angle at $w$ is at least
	$\slope(\overline{wv})$. In that case it follows that the maximum number of
	regions visited is $K[w,v]+I[w,v]$.

	If $w$ does have a predecessor $u$ in $S_v$, we have
	$\Lambda(v,w) = \Lambda(v,u) \cup W$, where $W$ is either empty or the
	singleton $\{w\}$. By Observation~\ref{obs:monotone_slopes} (and the
	definition of $K$) we have that
	$K[v,u] = \max_{u' \in \Lambda(v,u)}\max_{s \in \Lambda(u',v)} K[u',s] +
	I[u',v]$. Analogous to the base case we have
	$\max_{u' \in W} K[v,u'] \max_{s \in \Lambda(u',v)} K[u',s] + I[u',v] =
	\max\{1,K[w,v] + I[w,v]\}$. The lemma follows.
\end{proof}

Lemma~\ref{lem:recurrence_upper} now suggests a dynamic programming approach to
compute the $K[v,w]$ values for all pairs of bottom endpoints $v,w$: we process
the endpoints $v$ on increasing $x$-coordinate, and for each $v$, we compute
all $K[v,w]$ values in the order of $S_v$. To this end, we need to compute (i)
the (radial) orders $S_v$, for all bottom endpoints $v$, and (ii) the number of
regions intersected by a line segment $\overline{uv}$, for all pairs of bottom
endpoints $u$, $v$. We show that we can solve both these problems in $O(n^2)$
time. We then also obtain an $O(n^2)$ time algorithm to compute
$k^* = \max_{v,w} K[v,w]$.

\subparagraph{Computing predecessor slopes.} For each bottom endpoint $v$, we
simply sort the other bottom endpoints around $v$. This can be done in $O(n^2)$
time in total~\cite{overmars1988cyclic}\footnote{Alternatively, we can dualize
  the points into lines and use the dual arrangement to obtain all radial orders
  in $O(n^2)$ time.}. We can now obtain $S_v$ by splitting the resulting list
into two lists, one with all endpoints left of $v$ and one with the endpoints
right of $v$, and merging these lists appropriately. In total this takes
$O(n^2)$ time.

\subparagraph{Computing the number of intersections.} We use the standard
duality transform~\cite{bkos-cgaa-00} to map every point $p=(p_x,p_y)$ to a
line $p^* : y=p_x x - p_y$, and every non-vertical line $\ell : y=ax+b$ to a
point $\ell^*=(a,-b)$.
Consider the arrangement \A formed by the lines $p^*$
dual to all endpoints $p$ (both top and bottom) of all regions in \R. Observe
that in this dual space, a vertical line segment $R=\overline{pq} \in \RR$
corresponds to a strip $R^*$ bounded by two parallel lines $p^*$ and $q^*$. Let
$\RR^*$ denote this set of strips corresponding to \RR. It follows that if we
want to count the number of regions of \RR intersected by a query line $\ell$
we have to count the number of strips in $\RR^*$ containing the point $\ell^*$.

All our query segments $\overline{uv}$ are defined by two bottom endpoints $u$
and $v$, so the supporting line $\ell_{uv}$ of such a segment corresponds to a
vertex $\ell^*_{uv}$ of the arrangement \A. It is fairly easy to count, for
every vertex $\ell^*$ of \A, the number of strips that contain $\ell^*$, in a
total of $O(n^2)$ time; simply traverse each line of \A while maintaining the
number of strips that contain the current point.

Since in our case we wish to count only the regions intersected by a
line segment $\overline{uv}$ (rather than a line $\ell_{uv}$), we need two more
observations. Assume without loss of generality that $u_x < v_x$. This means we
wish to count only the strips $R^*$ that contain $\ell^*_{uv}$ and whose slope
$\slope(R)$ lies in the range $[u_x,v_x]$.

\begin{observation}
  \label{obs:top_and_bottom}
  Let $p^*$ be a line, oriented from left to right, and let $R^*$ be a
  strip. The line $p^*$ intersects the bottom boundary of $R^*$ before the top
  boundary of $R^*$ if and only if $\slope(p^*) > \slope(R^*)$.
\end{observation}

\begin{figure}[tb]
  \centering
  \includegraphics{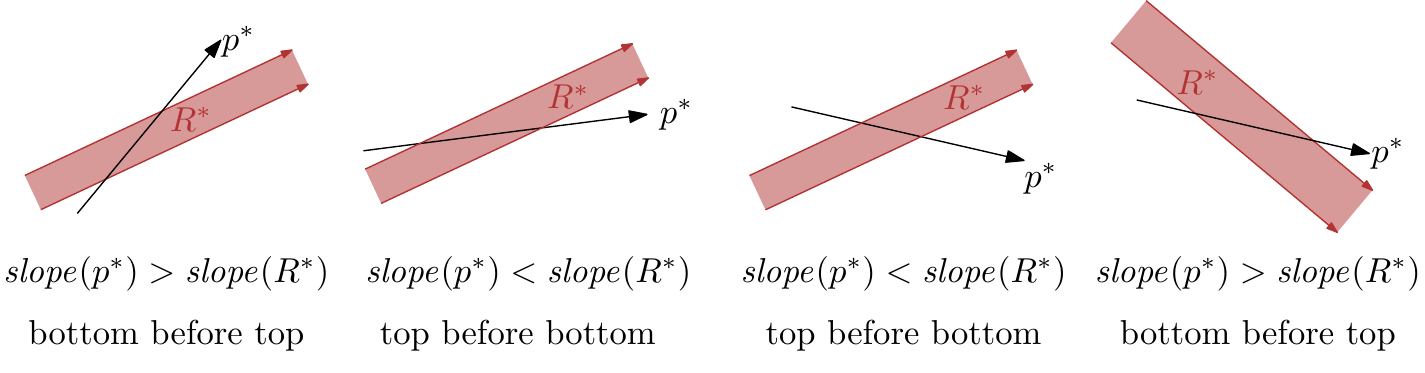}
  \caption{A line $p^*$ intersects the bottom of the strip $R^*$ if and only if
  $\slope(p^*) > \slope(R^*)$,}
  \label{fig:strips}
\end{figure}

Again consider traversing a line $p^*$ of \A (from left to right), and let
$T_{p^*}(\ell^*)$ be the number of strips that contain the point $\ell^*$ and that
we enter through the top boundary of the strip. 

\begin{lemma}
  \label{lem:strips_containing}
  Let $\ell^*_{uv}$, with $u_x < v_x$, be a vertex of \A at which the lines
  $u^*$ and $v^*$ intersect. The number of strips from $\RR^*$ with slope in
  the range $[u_x,v_x]$ containing $\ell^*_{uv}$ is
  $T_{u^*}(\ell^*_{uv}) - T_{v^*}(\ell^*_{uv})$.
\end{lemma}

\begin{proof}
  \begin{figure}[tb]
    \centering
    \includegraphics[page=2]{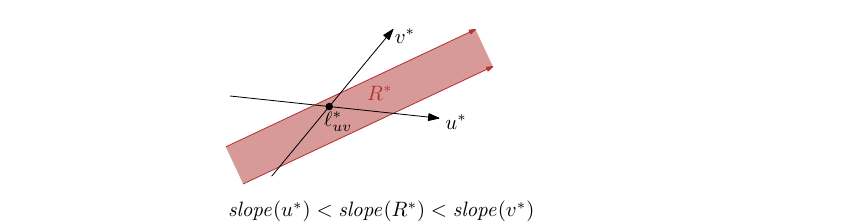}
    \caption{The strip $R^*$ with a slope in the range $[u_x,v_x]$ containing
      $\ell^*_{uv}$ contributes one to $T_{u^*}(\ell^*_{uv})$ and zero to
      $T_{v^*}(\ell^*_{uv})$.}
    \label{fig:strips_intersection}
  \end{figure}
  A strip that does not contain $\ell^*=\ell^*_{uv}$ contributes zero to both
  $T_{u^*}(\ell^*)$ and $T_{v^*}(\ell^*)$. A strip that contains $\ell^*$ but
  has slope larger than $v_x$ (and thus also larger than $u_x$) contributes one
  to both $T_{u^*}(\ell^*)$ and $T_{v^*}(\ell^*)$
  (Observation~\ref{obs:top_and_bottom}). Symmetrically, a strip that contains
  $\ell^*$ but has slope smaller than $u_x$ contributes zero to both
  $T_{u^*}(\ell^*)$ and $T_{v^*}(\ell^*)$. Finally, a strip whose slope is in
  the range $[u_x,v_x]$ is intersected by $u^*$ from the top, and by $v^*$ from
  the bottom (Observation~\ref{obs:top_and_bottom}), and thus contributes one
  to $T_{u*}(\ell^*)$ and zero to $T_{v^*}(\ell^*)$. See
  Figure~\ref{fig:strips_intersection} for an illustration. The lemma follows.
\end{proof}

\begin{corollary}
  \label{cor:segments_intersecting}
  Let $u, v \in \Lower\RR$ be bottom endpoints. The number of regions of \RR
  intersected by $\overline{uv}$ is
  $T_{u^*}(\ell^*_{uv}) - T_{v^*}(\ell^*_{uv})$.
\end{corollary}

We can easily compute the counts $T_{u*}(\ell^*_{uv})$ for every vertex
$\ell^*_{uv}$ on $u^*$ by traversing the line $u^*$. We therefore obtain the
following result.

\begin{lemma}
  \label{lem:intersection_counting}
  For all pairs of bottom endpoints $u,v \in \Lower\RR$, we can compute the
  number of regions in \RR intersected by $\overline{uv}$, in a total of
  $O(n^2)$ time.
\end{lemma}

Applying this in our dynamic programming approach for computing $k^*$ we get:

\begin{theorem}
  \label{thm:parallel_segments_upper_hull}
  Given a set of $n$ vertical line segments \RR, we can compute the maximum
  number of regions $k^*$ visitable by an upper convex transversal $Q$ in
  $O(n^2)$ time.
\end{theorem}

\subsection{Computing a convex transversal}
\label{sub:Computing_a_convex_transversal}

We now consider computing a convex partial transversal that maximizes the number
of regions visited. We first prove some properties of convex transversals.
We then use these properties to compute the maximum number of regions visitable
by such a transversal using dynamic programming.

\subsubsection{Canonical Transversals}
\label{ssub:Canonical_Transversals}

\begin{lemma}
  \label{lem:discrete_convex_hull_buildup}
  Let $Q$ be a convex partial transversal of \RR. There exists a convex partial
  transversal $Q'$ of \RR such that
  \begin{itemize}
  \item the transversals have the same leftmost vertex $\ell$ and
    the same rightmost vertex $r$,
  \item the upper hull of $Q'$ intersects the same regions as the upper hull of
    $Q$,
  \item all strictly convex vertices on the upper hull of $Q'$ lie on bottom
    endpoints of \RR,
  \item the lower hull of $Q'$ intersects the same regions as the lower hull of
    $Q$, and
  \item all strictly convex vertices on the lower hull of $Q'$ lie on top
    endpoints of regions in \RR.
  \end{itemize}
\end{lemma}

\begin{proof}
  Clip the segments containing $\ell$ and $r$ such that $\ell$ and $r$ are the
  bottom endpoints, and apply Lemma~\ref{lem:discrete_upper_hull} to get an
  upper convex transversal $U$ of \RR whose strictly convex vertices lie on
  bottom endpoints and that visits the same regions as the upper hull of
  $Q$. So we can replace the upper hull of $Q$ by $U$. Symmetrically, we can
  replace the lower hull of $Q$ by a transversal that visits the same regions and
  whose strictly convex vertices use only top endpoints.
\end{proof}

A partial convex transversal $Q'$ of \RR is a \emph{lower canonical} transversal if
and only if
\begin{itemize}[nosep]
\item the strictly convex vertices on the upper hull of $Q'$ lie on bottom
  endpoints in \RR,
\item the strictly convex vertices on the lower hull of $Q'$ lie on bottom or
  top endpoints  of regions in \RR,
\item the leftmost vertex $\ell$ of $Q'$ lies on a line through $w$, where
  $w$ is the leftmost strictly convex vertex of the lower hull of $Q'$, and
  another endpoint.
\item the rightmost vertex $r$ of $Q'$ lies on a line through $z$, where $z$
  is the rightmost strictly convex vertex of the lower hull of $Q'$, and
  another endpoint.
\end{itemize}
An \emph{upper canonical} transversal is defined analogously, but now $\ell$ and
$r$ lie on lines through an endpoint and the leftmost and rightmost strictly
convex vertices on the upper hull.

\begin{lemma}
  \label{lem:discrete_convex_hull}
  Let $Q$ be a convex partial transversal of \RR for which all $h\geq 2$ strictly
  convex vertices in the lower hull lie on endpoints of regions in \RR. There
  exists a lower canonical transversal $Q'$ of \RR, that visits the same regions
  as $Q$.
\end{lemma}

\begin{proof}
  \begin{figure}[tb]
    \centering
    \includegraphics{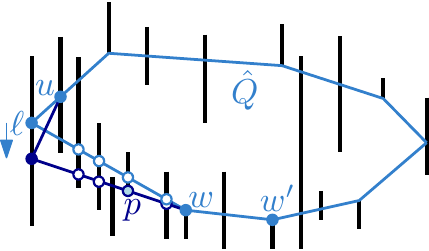}
    \caption{We move $\ell$ and the vertices on $\overline{\ell{}w}$
      downwards, bending at $u$ and $w$ until $\overline{\ell{}w}$ contains a
      bottom endpoint $p$ or becomes colinear with $\overline{ww'}$.}
    \label{fig:canonical_hull}
  \end{figure}
  Let $\ell$ be the leftmost point of $Q$, let $u$ be the vertex of $Q$
  adjacent to $\ell$ on the upper hull, let $w$ be the leftmost strictly convex
  vertex on the lower hull of $Q$, and let $w'$ be the strictly convex vertex
  of $Q$ adjacent to $w$. We move $\ell$ and all other vertices of $Q$ on
  $\overline{\ell{}w}$ downwards, bending at $u$ and $w$, while $Q$ remains
  convex and visits the same $k$-regions until: (i) $\ell$ lies on the bottom
  endpoint of its segment, (ii) the segment $\overline{\ell{}w}$ contains the
  bottom endpoint $p$ of a region, or (iii) the segment $\overline{\ell{}w}$
  has become collinear with $\overline{ww'}$. See
  Figure~\ref{fig:canonical_hull}. Observe that in all cases $\ell$ lies on a
  line through the leftmost strictly convex vertex $w$ and another endpoint
  (either $\ell$ itself, $p$, or $w'$). Symmetrically, we move $r$ downwards
  until it lies on an endpoint or on a line through the rightmost strictly
  convex vertex $z$ on the lower hull of $Q$ and another endpoint. Let $Q''$ be
  the resulting convex transversal we obtain.

  Let $\RR'$ be the regions intersected by the upper hull of $Q''$ (setting the
  bottom endpoint of the regions containing $\ell$ and $r$ to be $\ell$ and
  $r$). We now appeal to Lemma~\ref{lem:discrete_upper_hull} to get that there
  is an upper hull that also visits all regions in $\RR'$ and in which all
  strictly convex vertices lie on bottom endpoints. So, we can replace the
  upper hull of $Q''$ by this upper hull and obtain the transversal $Q'$ stated
  in the lemma.
\end{proof}

\begin{lemma}
  \label{lem:discrete_straight}
  Let $Q$ be a convex partial transversal of \RR, whose lower hull intersects at
  least one region (other than the regions containing the leftmost and
  rightmost vertices) but contains no endpoints of regions in \RR. There exists
  a convex partial transversal $Q'$ intersecting the same regions as $Q$ that
  does visit one endpoint of a region in \RR.
\end{lemma}
\begin{proof}
  Consider the region $R$ intersected by $\overline{\ell{}r}$ whose bottom
  endpoint $p$ minimizes the distance between $p$ and
  $w=R\cap\overline{\ell{}r}$, and observe that $w$ is a vertex of $Q$. Shift
  down $w$ (and the vertices on $\overline{\ell{}w}$ and $\overline{wr}$) until
  $w$ lies on $p$.
\end{proof}

Let $Q=\ell{}urv$ be a quadrilateral whose leftmost vertex is $\ell$, whose
\emph{top} vertex is $u$, whose rightmost vertex is $r$, and whose
\emph{bottom} vertex is $v$. The quadrilateral $Q$ is a \emph{lower canonical}
quadrilateral if and only if
\begin{itemize}[nosep]
\item $u$ and $v$ lie on endpoints in $\Both\RR$,
\item $\ell$ lies on a line through $v$ and another endpoint, and
\item $r$ lies on a line through $v$ and another endpoint.
\end{itemize}
We define \emph{upper canonical} quadrilateral analogously (i.e. by requiring
that $\ell$ and $r$ lie on lines through $u$ rather than $v$).

\begin{lemma}
  \label{lem:quadrilateral}
  Let $Q=\ell{}urv$ be a convex quadrilateral with $\ell$ as leftmost vertex,
  $r$ as rightmost vertex, $u$ a bottom endpoint on the upper hull of $Q$, and
  $v$ an endpoint of a region in \RR.
  \begin{itemize}
  \item There exists an upper or lower canonical quadrilateral $Q'$ intersecting
    the same regions as $Q$, or
  \item there exists a convex partial transversal $Q''$ whose upper hull contains
    exactly two strictly convex vertices, both on endpoints, or whose lower
    hull contains exactly two strictly convex vertices, both on endpoints.
  \end{itemize}
\end{lemma}

\begin{proof}
  \begin{figure}[tb]
    \centering
    \includegraphics[width=\textwidth]{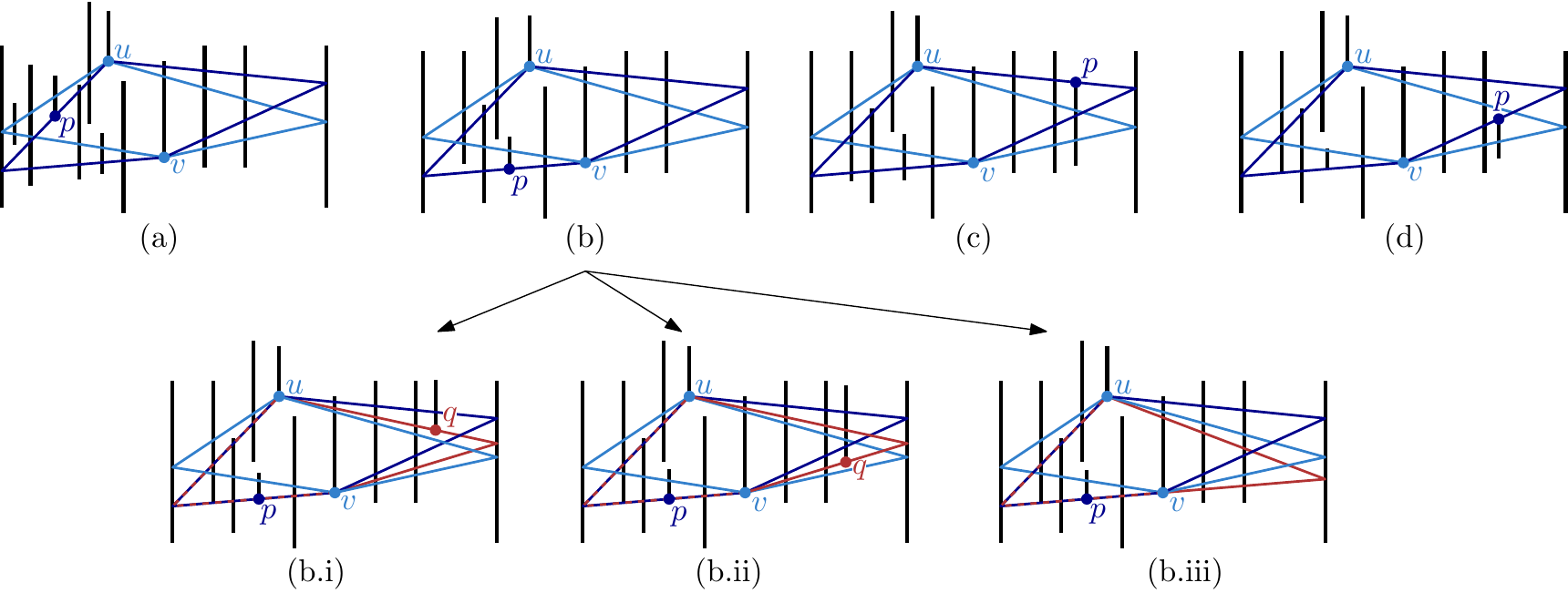}
    \caption{We can transform a quadrilateral transversal $Q$ into a transversal
      with two strictly convex vertices on the lower hull or on the upper hull
      (cases (a) and (d)), or into a quadrilateral in which the leftmost vertex
      and rightmost vertex lie on a line through an endpoint and either $u$ or
      $v$ (cases (b) and (c)). The bottom row shows how we can
      shift $r$ to lie on a line through $v$ and another endpoint $q$ when we
      already have that $\ell$ lies on a line through $v$ and an endpoint $p$.}
    \label{fig:quadrilateral_cases}
  \end{figure}

  Shift $\ell$ downward and $r$ upward with the same speed, bending at $u$ and
  $v$ until one of the edges of the quadrilateral will stop to intersect the
  regions only intersected by that edge. It follows that at such a time this
  edge contains an endpoint $p$ of a region in \RR. We now distinguish between
  four cases, depending on the edge containing $p$. See
  Figure~\ref{fig:quadrilateral_cases}.

  \begin{enumerate}[label=(\alph*)]
  \item Case $p \in \overline{\ell{}u}$. It follows $p$ is a bottom
    endpoint. We continue shifting $\ell$ and $r$, now bending in $p$, $u$, and
    $v$. We now have a convex partial transversal $Q''$ that visits the same
    regions as $Q$ and whose upper hull contains at least two strictly convex
    vertices, both bottom endpoints of regions.
  \item Case $p \in \overline{\ell{}v}$. It follows that $p$ is a bottom
    endpoint. We now shift $r$ downwards until: (b.i) $\overline{ur}$ contains a
    bottom endpoint $q$, (b.ii) $\overline{v{}r}$ contains a bottom endpoint $q$
    or (b.iii) $\overline{vr}$ is collinear with $\overline{\ell{}v}$. In the case
    (b.i) we get an upper hull with two strictly convex vertices, both bottom
    endpoints of regions. In cases (b.ii) and (b.iii) we now have that both $\ell$
    and $r$ lie on lines through $v$ and another endpoint.
  \item Case $p \in \overline{ur}$. Similar to the case
    $p \in \overline{\ell{}v}$ we now shift $\ell$ back upwards until the lower
    hull contains at least two strictly convex endpoints, or until $\ell$ lies
    on a line through $u$ and some other endpoint $q$.
  \item Case $p \in \overline{vr}$. Similar to the case
    $p \in \overline{\ell{}u}$. We continue shifting $\ell$ downward and $r$
    upward, bending in $p$, thus giving two strictly convex vertices in the
    lower hull, both on endpoints.
  \end{enumerate}
\end{proof}

Let $k^u_2$ be the maximal number of regions of \RR visitable by an upper
convex transversal, let $k^u_4$ be the maximal number of regions of \RR visitable
by a canonical upper quadrilateral, and let $k^u$ denote the maximal number of
regions of \RR visitable by a canonical upper transversal. We define $k^b_2$, $k^b_4$, and
$k^b$, for the maximal number of regions of \RR, visitable by a lower convex
transversal, canonical lower quadrilateral, and canonical lower transversal,
respectively.

\begin{lemma}
  \label{lem:optimal_convex_transversal}
  Let $k^*$ be the maximal number of regions in \RR visitable by a convex partial
  transversal of \RR. We have that $k^*=\max\{k^u_2,k^u_4,k^u,k^b_2,k^b_4,k^b\}$.
\end{lemma}

\begin{proof}
  Clearly $k^* \geq \max\{k^u_2,k^u_4,k^u,k^b_2,k^b_4,k^b\}$. We now argue that
  we can transform an optimal convex partial transversal $Q^*$ of \RR visiting
  $k^*$ regions into a canonical transversal. The lemma then follows.

  By Lemma~\ref{lem:discrete_convex_hull_buildup} there is an optimal convex
  partial transversal $Q$ of \RR, visiting $k^*$ regions, whose strictly convex
  vertices lie on endpoints of the regions in \RR.

  If either the lower or upper hull of $Q$ does not intersect any regions, we
  use Lemma~\ref{lem:discrete_upper_hull} (or its analog for the bottom hull)
  to get $k^*=\max\{k^u_2,k^b_2\}$. Otherwise, if the lower hull of $Q$ contains
  at least two strictly convex vertices, we apply
  Lemma~\ref{lem:discrete_convex_hull}, and obtain that there is a convex
  lower transversal. Similarly, if the upper hull contains at least two strictly
  convex vertices we apply a lemma analogous to
  Lemma~\ref{lem:discrete_convex_hull}, and obtain that there is a canonical
  upper transversal visiting the same $k^*$ regions as $Q$.

  If $Q$ has at most one strictly convex vertex on both the lower hull and
  upper hull we use Lemma~\ref{lem:discrete_straight} and get that there exists
  a convex quadrilateral $Q'$ that visits the same regions as $Q$. We now apply
  Lemma~\ref{lem:quadrilateral} to get that there either is an optimal
  transversal that contains two strictly convex vertices on its upper or lower
  hull, or there is a canonical quadrilateral $Q'$ that intersects the same
  regions as $Q$. In the former case we can again apply
  Lemma~\ref{lem:discrete_convex_hull} to get a canonical convex partial
  transversal that visits $k^*$ regions. In the latter case we have
  $k^*=\max\{k^u_4,k^b_4\}$.
\end{proof}

By Lemma~\ref{lem:optimal_convex_transversal} we can restrict our attention to
upper and lower convex transversals, canonical quadrilaterals, and canonical
transversals. We can compute an optimal upper (lower) convex transversal in
$O(n^2)$ time using the algorithm from the previous section. Next, we argue that
we can compute an optimal canonical quadrilateral in $O(n^5)$ time, and an
optimal canonical transversal in $O(n^6)$ time. Arkin \etal~\cite{ARKIN2014224}
describe an algorithm that given a discrete set of vertex locations can find a
convex polygon (on these locations) that maximizes the number of regions
stabbed. Note, however, that since a region contains multiple vertex locations
---and we may use only one of them--- we cannot directly apply their algorithm.

Observe that if we fix the leftmost strictly convex vertex $w$ in the lower
hull, there are only $O(n^2)$ candidate points for the leftmost vertex $\ell$
in the transversal. Namely, the intersection points of the (linearly many)
segments with the linearly many lines through $w$ and an endpoint. Let $L(w)$
be this set of candidate points. Analogously, let $R(z)$ be the candidate
points for the rightmost vertex $r$ defined by the rightmost strictly convex
vertex $z$ in the lower hull.

\subsection{Computing the maximal number of regions intersected by a
	canonical quadrilateral}

Let $Q=\ell{}urw$ be a canonical lower quadrilateral with $u_x < w_x$, let
$\RR''$ be the regions intersected by $\overline{u\ell}\cup\overline{\ell{}w}$,
and let $T[w,u,\ell] = |\RR''|$ be the number of such regions. We then define
$\RR_{uw\ell}=\RR \setminus \RR''$ to be the remaining regions, and
$I(\RR_{uw\ell},u,w,r)$ as the number of regions from $\RR$ intersected by
$\overline{ur}\cup\overline{wr}$. Observe that those are the regions from $\RR$
that are \emph{not} intersected by $\overline{u\ell}\cup\overline{\ell{}w}$.
Note that we exclude the two regions that have $u$ or $w$ as its endpoint from
$\RR_{uw\ell}$.  Hence, the number of regions intersected by $Q$ is
$T[w,u,\ell] + I(\RR_{uw\ell},u,w,r)$, See
Figure~\ref{fig:1-oriented-canonical-quadrilateral}, and the maximum number of
regions all canonical lower quadrilaterals with $u_x \leq w_x$ is
\[ \max_{u,w,\ell} (T[u,w,\ell] + \max_{r \in R(w)} I(\RR_{uw\ell},u,w,r)).
\]
We show that we can compute this in $O(n^5)$ time. If $u_x > w_x$, we use a
symmetric procedure in which we count all regions intersected by
$\overline{ur}\cup\overline{rw}$ first, and then the remaining regions
intersected by $\overline{u\ell}\cup\overline{\ell{}w}$. Since $k^b_4$ is the
maximum of these two results, computing $k^b_4$ takes $O(n^5)$ time as well.

Since there are $O(n)$ choices for $w$ and $u$, and $O(n^2)$ for $\ell$, we can
naively compute  $T[w,u,\ell]$ and $\RR_{uw\ell}$ in $O(n^5)$ time. For each
$w$, we then radially sort $R(w)$ around $w$. Next, we describe how we can then
compute all $I(\RR',u,w,r)$, with $r \in R(w)$, for a given set
$\RR'=\RR_{uw\ell}$, and how to compute $\max_r I(\RR',u,w,r)$.

\begin{figure}
	\centering
	\includegraphics{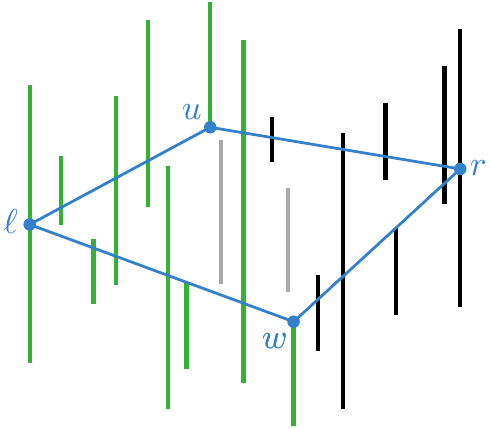}
	\caption{
		The green vertical segments are $\RR''$, while the gray and black ones are $\RR_{uw\ell}$.
			The number of black vertical segments is $I(\RR_{uw\ell},u,w,r)$.
    }
	\label{fig:1-oriented-canonical-quadrilateral}
\end{figure}


\subparagraph{Computing the number of Intersections} Given a set of regions
$\RR'$, the points $u$, $w$, and the candidate endpoints $R(w)$, sorted
radially around $w$, we now show how to compute $I(\RR',u,w,r)$ for all
$r \in R(w)$ in $O(n^2)$ time.

\begin{figure}[tb]
	\centering
	\includegraphics{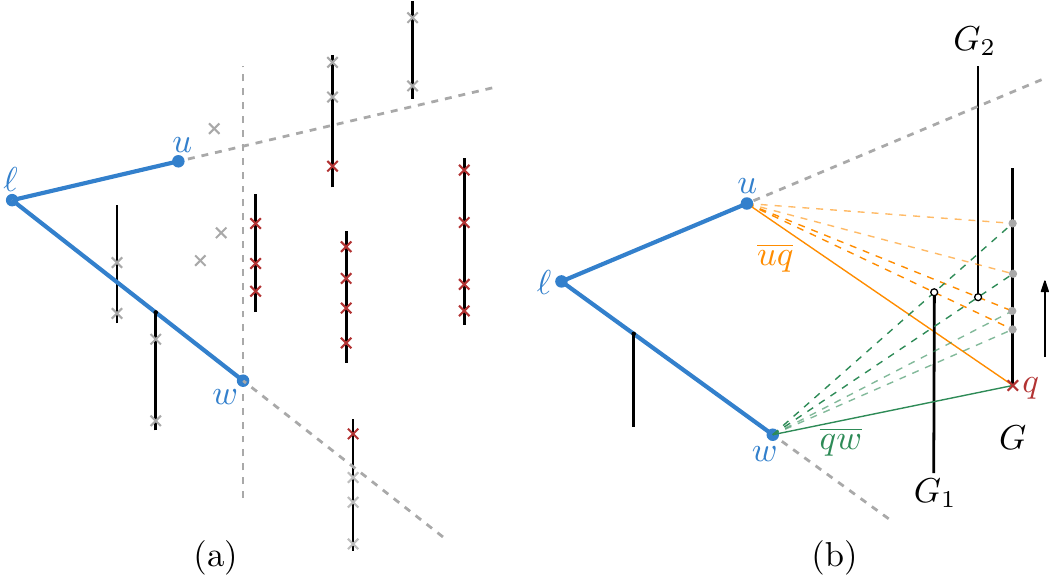}
	\caption{(a) The red crosses are possible candidate points in $R(w)$.
	Gray crosses above $\overline{u\ell}$ or below $\overline{w\ell}$ or left to the vertical line through $w$ is not part of $R(w)$.
	(b) The orange line segment $\overline{uq}$ and green line segment $\overline{qw}$ move upwards along $G$.
	An events occurs when either $\overline{uq}$ or $\overline{qw}$ sweeps through the endpoint of a vertical segment in $\R''$.
	When $\overline{uq}$ sweeps through the top endpoint of $G_1$, we check whether $G_1$ intersects with $\overline{qw}$ or not.
	We increase the number of intersected segment when $\overline{uq}$ sweeps through the bottom endpoint of $G_2$ and decrease
	it when $\overline{qw}$ sweeps through the top endpoint of $G_1$.}
	\label{fig:1-oriented-quadrilateral_closing}
\end{figure}

We sort the endpoints of $\RR'$ radially around $w$ and around $u$, and
partition the points in $R(w)$ based on which segment (region) they lie. Let
$G$ be a region that has $q$ as its bottom endpoint, and let
$R(w,G) = R(w)\cap G$. We explicitly compute the number of regions $m$
intersected by $\overline{uq}\cup\overline{qw}$ in linear time, and set
$I(\RR',u,w,q)$ to $m$. We shift up $q$, while maintaining the number of
regions intersected by $\overline{uq}\cup\overline{qw}$. See
Figure~\ref{fig:1-oriented-quadrilateral_closing}(b)
As $q$ moves
upwards, the segments $\overline{uq}$ and $\overline{qw}$, sweep through
endpoints of the regions in $\RR'$. We can decide which event should be processed
first from the two ordered lists of endpoints in constant time.

If an event occurs, the value of $m$ changes depending on the type of endpoint
and which line segment ($\overline{uq}$ or $\overline{qw}$) responsible for
that particular event.  If $\overline{uq}$ sweeps through a point in
$\Lower \RR$, we increase the value of $m$ by one.  Conversely, we decrease the
value of $m$ by one if $\overline{uq}$ sweeps a top endpoint of some region
$G'$ if $G'$ does not intersect with $\overline{qw}$.  Otherwise, we do
nothing.  We treat events caused by $\overline{qw}$ in a similar way. See
Figure~\ref{fig:1-oriented-quadrilateral_closing}(b) for an illustration. When $q$ passes
through a candidate point $r \in R(w,G)$ we set $I(\RR',u,w,r)$ to
$m$. Computing all values $I(\RR,u,w,r)$ for all $r \in R(w,G)$ then takes
$O(n)$ time, and thus $O(n^2)$ time over all regions $G$.

\subparagraph{Maximizing over $r$.} To find the point $r \in R(w)$ that
maximizes $R'[u,w,\ell,r]$ we now just filter $R(w)$ to exclude all points
above the line through $u$ and $\ell$ and below the line through $w$ and
$\ell$, and report the maximum $I$ value among the remaining points.


\subparagraph{Improving the running time to $O(n^5)$.} Directly applying the
above approach yields an $O(n^6)$ time algorithm, as there are a total of
$O(n^4)$ triples $(u,w,\ell)$ to consider. We now improve this to $O(n^5)$ time
as follows.

First observe that since $u_x < w_x$, $\overline{ur}\cup\overline{rw}$ cannot
intersect any regions of the regions left of $u$. Let $\RR_u$ be this set of
regions. Since $\ell$ lies left of $u$ (by definition) we thus have that
$I(\RR_{uw\ell},u,w,r) = I(\RR_{uw\ell}\setminus \RR_u,u,w,r)$.

Second, consider all points $\ell \in L(w)$ that lie on the line through $w$
and some endpoint $p$. Observe that they all have the same set
$\RR_{uw\ell}\setminus \RR_u$. This means that there are only $O(n^3)$
different sets $\RR_{uwp}=\RR_{uw\ell}\setminus \RR_u$ for which we have to
compute $I(\RR_{uwp},u,w,r)$ values.

Finally, again consider all points $\ell \in L(w)$ that lie on the line $\mu$
through $w$ and some point $p$. For all these points we discard the same points
from $R(w)$ because they are below $\mu$. We can then compute
$T[u,w,\ell] + \max_r I(\RR_{uwp},u,w,r)$ for all these points $\ell$ in
$O(n^2)$ time in total, by rotating the line $\nu$ through $u$ and $\ell$
around $u$ in counter clockwise order, while maintaining the valid candidate
points in $R(w)$ (i.e. above $\mu$ and below $\nu$), and the maximum
$I(\RR_{uwp},u,w,r)$ value among those points,
see Figure~\ref{fig:1-oriented-improve-quadrilateral_closing}.
Since there are $O(n^3)$
combination of $w,u$, and $p$, then in total we spent $O(n^5)$ time to compute
all values of $T[u,w,\ell] + \max_r I(\RR_{uwp},u,w,r)$, over all $u$, $w$, and
$\ell$, and thus we obtain the following result.

\begin{figure}[tb]
	\centering
	\includegraphics{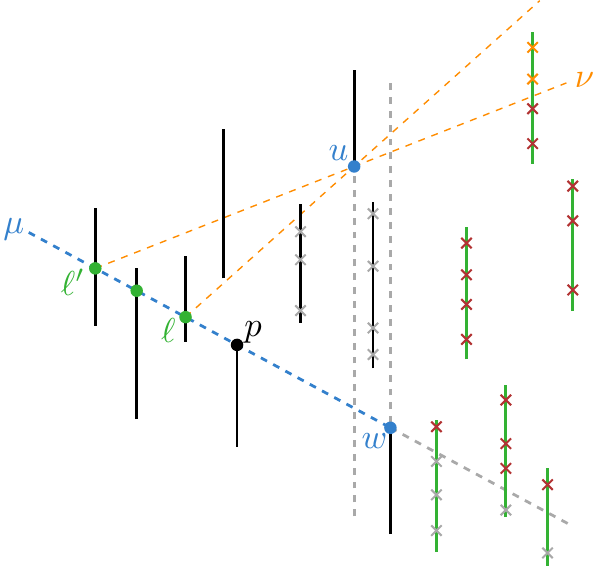}
	\caption{
		$\RR_{uwp}$ is shown in green vertical segments.
		The red crosses below $\nu$ (that has one of its endpoint at $\ell$'),
		above $\mu$ and to the left of vertical line through $w$ are candidates for the rightmost point.
		The orange crosses become candidates for the rightmost point after we rotate $\nu$ in
		counter-clockwise order and consider $\ell$ as the leftmost point.
		}
	\label{fig:1-oriented-improve-quadrilateral_closing}
\end{figure}

\begin{lemma}
  \label{lem:compute_canonical_quadrilateral}
  Given a set of $n$ vertical line segments \RR, we can compute the maximum
  number of regions $k^*$ visitable by a canonical quadrilateral $Q$ in
  $O(n^5)$ time.
\end{lemma}

\subsubsection{Computing the maximal number of regions intersected by a
  canonical transversal}

Next, we describe an algorithm to compute the maximal number of regions
visitable by a lower canonical convex transversal. Our algorithm consists of
three dynamic programming phases, in which we consider (partial) convex hulls
of a particular ``shape''.

In the first phase we compute (and memorize) $B[w,u,v,\ell]$: the maximal number of regions
visitable by a partial transversal that has $\overline{w\ell{}}$ as a segment
in the lower hull, and a
convex chain $\ell,\dots,u,v$ as upper hull. See Figure~\ref{fig:convexhull_dp}(a).

In the second phase we compute $K[u,v,w,z]$: the maximal number of regions
visitable by the canonical partial convex transversal whose rightmost top edge is
$\overline{uv}$ and whose rightmost bottom edge is $\overline{wz}$.

In the third phase we compute the maximal number of regions visitable when we ``close''
the transversal using the rightmost vertex $r$. To this end, we define
$R'[z,u,v,r]$ as the number of regions visitable by the canonical transversal
whose rightmost upper segment is $\overline{uv}$ and whose rightmost bottom
segment is $\overline{wz}$ and $r$ is defined by the strictly convex vertex $z$.

\subparagraph{Computing $B[w,u,v,\ell]$.}
Given a set of regions $\RR'$ let $U_{\RR'}[\ell,u,v]$ be the maximal number of
regions in $\RR'$ visitable by an upper convex transversal that starts in a fixed
point $\ell$ and ends with the segment $\overline{uv}$.

\begin{lemma}
  \label{lem:compute_U}
  We can compute all values $U_{\RR'}[\ell,u,v]$ in $O(n^2)$ time.
\end{lemma}

\begin{proof}
  Analogous to the algorithm in Section~\ref{sub:Computing_an_upper_convex_transversal}.
\end{proof}

Let $B[w,u,v,\ell]$ be the maximal number of regions visitable by a transversal
that starts at $w$, back to $\ell$, and then an upper hull from $\ell$ ending
with the segment $\overline{uv}$. See Figure~\ref{fig:convexhull_dp}(a).

\begin{figure}[tb]
  \centering
  \includegraphics{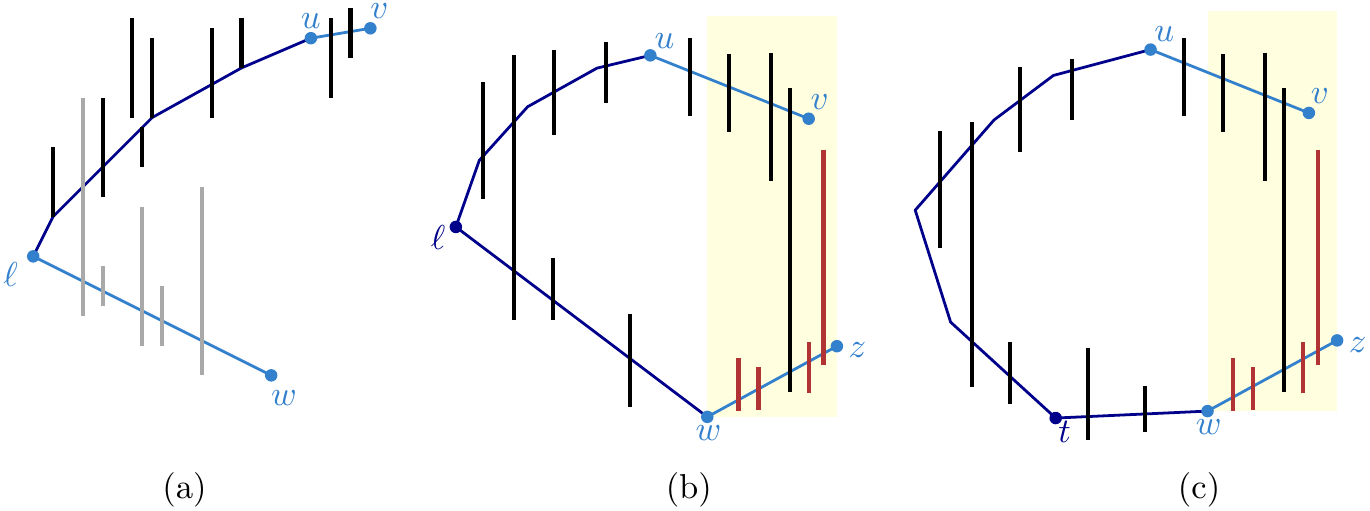}
  \caption{(a) $B[w,u,v,\ell]$ indicates the number of regions visited by a
    partial convex transversal that has $\overline{\ell{}w}$ as bottom hull and
    the upper hull from $\ell$ to $\overline{uv}$. We can compute the
    $B[w,u,v,\ell]$ values for all $u,v$ by explicitly setting aside the
    segments intersected by $\overline{\ell{}w}$ and then using the upper hull
    algorithm. (b) The base case of the recurrence when $u_x < w_x$. The
    regions counted by $I[w,z,u,v]$ are shown in red, whereas the regions
    counted by $B[w,u,v,\ell]$ are shown in black. (c) The inductive step when
    $u_x < w_x$.}
  \label{fig:convexhull_dp}
\end{figure}

For each combination of $w$ and $\ell \in L(w)$, explicitly construct the
regions $\RR^{w,\ell}$ \emph{not} intersected by $\overline{w\ell}$. We then
have that
$B[w,u,v,\ell] = |\RR \setminus \RR^{w,\ell}| +
U_{\RR^{w,\ell}}[\ell,u,v]$. So, we can compute all $B[w,u,v,\ell]$, for
$w \in \Both\RR, \ell \in W(\ell)$, and $u,v \in \Upper\RR$, in $O(n^5)$ time
in total.

\subparagraph{Computing $X[a,u,w,z,\ell]$.} Let $X[a,u,w,z,\ell]$ be the
maximum number of regions \emph{left of $u$} visitable by a partial transversal
that
\begin{itemize}[nosep]
\item has $a$ as its leftmost strictly convex vertex in the bottom hull,
\item has $\ell \in L(a)$ as its leftmost vertex,
\item has $\overline{wz}$ as rightmost edge in the bottom hull, and
\item has $\overline{\ell{}u}$ as its (partial) upper hull.
\end{itemize}

We can compute $X[a,u,w,z,\ell]$ using an approach similar to the one we used to
compute $B[w,u,v,\ell]$: we fix $a, u$, and $\ell$, and explicitly compute the
segments intersected by $\overline{\ell{}u}\cup\overline{\ell{}a}$. We set
those we count them, and set them aside. On the remaining segments $\RR'$ we compute
an optimal bottom hull from $a$ to $z$ whose incoming slope at $a$ is at least
$\slope(\ell,a)$. This takes $O(n^2)$ time, using an analogous approach to that used in
Section~\ref{sub:Computing_an_upper_convex_transversal}. Since we have $O(n^4)$ choices for the triple $a,u,\ell$ this
takes a total of $O(n^6)$ time.



\subparagraph{Computing $K[u,v,w,z]$.}
Let $\overleftarrow{\triangle}(P,u,v,w,z)$ denote the subset of the points $P$ left
of $\{u,v,w,z\}$, below the line through $u$ and $v$, and above the line
through $w$ and $z$.

Let $K[u,v,w,z]$ be a maximal number of regions visitable by the minimum area
canonical partial convex transversal that has $\overline{uv}$ as its rightmost
segment in the upper hull and $\overline{wz}$ as its rightmost segment in the
lower hull.

If $u_x < w_x$ we have that%
\begin{align*}
  K[u,v,w,z] = \max\{&\max_{\ell \in \overleftarrow{\triangle}(L(w),u,v,w,z)}      B[w,u,v,\ell] + I[w,z,u,v], \\
                     &\max_{t \in \overleftarrow{\triangle}(\Both\RR,u,v,w,z)} K[u,v,t,w]    + I[w,z,u,v]\},
\end{align*}
where $I[w,z,u,v]$ is the number of segments intersected by $\overline{wz}$ but
not by $\overline{uv}$. See Figure~\ref{fig:convexhull_dp}(b) and (c) for an
illustration. We rewrite this to
\[ K[u,v,w,z] = I[w,z,u,v] + \max\{\max_{\ell \in \overleftarrow{\triangle}(L(w),u,v,w,z)} B[w,u,v,\ell]
                                  ,\max_{t \in \overleftarrow{\triangle}(\Both\RR,u,v,w,z)} K[u,v,t,w]\}. \]

\begin{figure}[tb]
  \centering
  \includegraphics{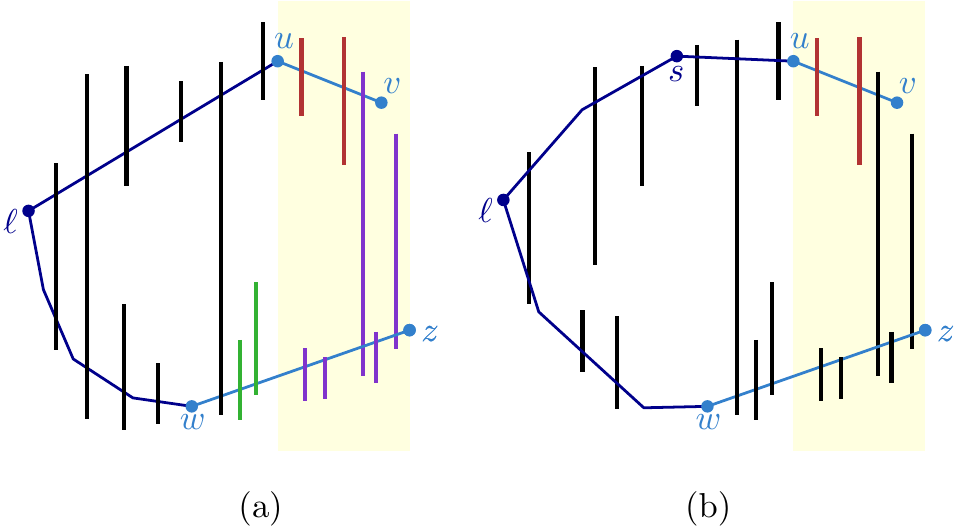}
  \caption{The two cases in the dynamic program when $u_x > w_x$. (a) The
    segments counted by $T[w,u,\ell]$ in black, $J[w,z,\ell,u]$ in green,
    $I[u,v,w,z]$ in red, and $I^*[w,z,u]$ in purple. (b) The case where there
    is another segment in the lower hull as well. The segments counted by
    $I[u,v,w,z]$ are again shown in red, whereas the segments counted by
    $K[s,u,w,z]$ are shown in black.}
  \label{fig:convexhull_dp_wu}
\end{figure}

If $u_x > w_x$ we get a more complicated expression.  Let $J[w,z,\ell,u]$
be the number of regions intersected by $\overline{wz}$ but not by
$\overline{\ell{}u}$, and let $I^*[w,z,u]$ be the number of regions right of
$u$ intersected by $\overline{wz}$. See Figure~\ref{fig:convexhull_dp_wu} for an
illustration. We then have
\begin{align*}
  K[u,v,w,z] = \max&\{\max_{\substack{a \in \overleftarrow{\triangle}(\Lower\RR,u,v,w,z),\\
                     \ell \in L(a)}} X[a,u,w,z,\ell] + I^*[w,z,u] + I[u,v,w,z]\\
                     &, \max_{s \in \overleftarrow{\triangle}(\Lower\RR,u,v,w,z)} K[s,u,w,z]    + I[u,v,w,z]\}
\end{align*}
which we rewrite to
\begin{align*}
  K[u,v,w,z] = I[u,v,w,z] + \max\{&I^*[w,z,u] + \max_{\substack{a \in \overleftarrow{\triangle}(\Lower\RR,u,v,w,z),\\
                     \ell \in L(a)}} X[a,u,w,z,\ell]\\
                     &, \max_{s \in \overleftarrow{\triangle}(\Lower\RR,u,v,w,z)} K[s,u,w,z]\}.
\end{align*}

We can naively compute all $I[w,z,u,v]$ and all $I^*[w,z,u]$ values in $O(n^5)$
time. Since there are only $O(n)$ choices for $t$ computing all values
$\max_t K[u,v,t,w]$ for all cells takes only $O(n^5)$ time in total. The same
holds for computing all $\max_s K[s,u,w,z]$.

As we describe next, we can compute the $\max_\ell B[w,u,v,\ell]$ values in
$O(n^5)$ time as well. Fix $w$, compute all $O(n^2)$ candidate points in $L(w)$
and sort them radially around $w$. For each $\overline{uv}$: remove the
candidate points above the line through $u$ and $v$. This takes $O(n^5)$ time
in total.

For each maximal subset $S \subseteq L(w)$ that lies below the line through $u$ and
$v$, we now do the following. We rotate a line around $w$ in counterclockwise
order, starting with the vertical line. We maintain the subset of $S$ that lies
above this line, and $\max_{\ell \in S} B[w,u,v,\ell]$. Thus, when this line
sweeps through a point $z$, we know the value
$\max_{\ell \in \overleftarrow{\triangle}(L(w),u,v,w,z)} B[w,u,v,\ell]$. There
are $O(n^3)$ sets $S$, each of size $O(n^2)$. So, while rotating the line we
process $O(n^2+n)=O(n^2)$ events. Since we can look up the value
$B[w,u,v,\ell]$ values in constant time, processing each event takes only
$O(1)$ time. It follows that we spend $O(n^5)$ time in total.

Similarly, computing all $\max_{a,\ell} X[a,u,w,z,\ell]$ values requires
$O(n^6)$ time: we fix $a,u,w,z$, and $\ell$, and compute all $O(n^6)$ values
$X[a,u,w,z,\ell]$ in $O(n^6)$ time. We group these values based on the slope of
$\overline{\ell{}a}$, and sort the groups on this slope. This again takes
$O(n^6)$ in total. Similarly, we sort the vertices $v$ around $u$, and
simultaneously scan the list of $X$ values and the vertices $v$, while
maintaining the maximum $X[a,u,w,z,\ell]$ value that has slope at least
$\slope(u,v)$. This takes $O(n^6)$ time.

It follows that we can compute all $K[u,v,w,z]$ values in $O(n^6)$ time in
total.

\subparagraph{Closing the hull.}  We now consider adding the rightmost point to
finish the convex transversal. Given vertices $u, v,w$, and $z$, let $\RR'$ be
the segments intersected by the partial transversal corresponding to
$K[u,v,w,z]$. We explicitly compute $\RR_{uvwz}=\RR \setminus \RR'$, and then compute the
maximum number of regions from this set intersected by
$\overline{vr}\cup\overline{rz}$, over all choices of rightmost vertex
$r$. This takes $O(n^2)$ time using the exact same approach we used in the
canonical quadrilateral section. It follows that we can compute the maximum
number of regions visitable by a canonical bottom convex transversal in
$O(n^6)$ time. Therefore we conclude:


\begin{theorem}
  \label{thm:parallel_segments_convex_hull}
  Given a set of $n$ vertical line segments \RR, we can compute the maximum
  number of regions $k^*$ visitable by a convex partial transversal $Q$ in
  $O(n^6)$ time.
\end{theorem}

\section{2-oriented disjoint line segments}
\label{sec:2-oriented}

In this section we consider the case when $\RR$ consists of vertical and horizontal disjoint segments. We will show how to apply similar ideas presented in previous sections to compute an optimal convex transversal $Q$ of $\RR$. As in the previous section, we will mostly restrict our search to canonical transversals. However, unlike in the one-oriented case, we will have one special case to consider when an optimal partial convex transversal has bends not necessarily belonging to a discrete set of points.

We call the left-, right-, top- and bottommost vertices $\ell$, $r$, $u$ and $b$ of a convex partial transversal the \emph{extreme} vertices. Consider a convex hull of a partial transversal $Q$, and consider the four convex chains between the extreme vertices. Let us call the chain between vertices $\ell$ and $u$ the \emph{upper-left hull}, and the other chains \emph{upper-right}, \emph{lower-right} and \emph{lower-left}. Similar to Lemma~\ref{lem:discrete_upper_hull} we can show the following:

\begin{lemma}\label{lem:2orient-canonical-1}
    Let $Q$ be a convex partial transversal of \RR with extreme vertices $\ell$, $r$, $u$, and $b$. There exists a convex partial transversal $Q'$ of \RR such that
    \begin{itemize}
        \item the two transversals have the same extreme vertices,
        \item all segments that are intersected by the upper-left, upper-right, lower-right, and lower-left hulls of $Q$ are also intersected by the corresponding hulls of $Q'$,
        \item all strictly convex vertices on the \textbf{upper-left} hull of $Q'$ lie on \textbf{bottom} endpoints of vertical segments or on the \textbf{right} endpoints of horizontal segments of \RR,
        \item the convex vertices on the other hulls of $Q'$ lie on analogous endpoints.
    \end{itemize}
\end{lemma}

One condition for a transversal to be canonical will be that all its strictly convex vertices, except for the extreme ones, satisfy the conditions of Lemma~\ref{lem:2orient-canonical-1}. Another condition will be that its extreme vertices belong to a discrete set of \emph{fixed} points. This set of fixed points will contain all the endpoints of the segments in \RR; we will call these points \emph{$0$th-order fixed points}. Furthermore, the set of fixed points will contain intersections of the segments of \RR of certain lines, that we will describe below, with the segments of \RR.

Now, consider a convex partial transversal $Q$ for which all the strictly convex vertices, except for the four extreme ones, lie on endpoints of segments in \RR. We will describe how to slide the extreme vertices of $Q$ along their respective segments to obtain a canonical transversal. Note, that for simplicity of exposition in the description below we assume that no new intersections of the convex hull of $Q$ with segments of \RR appear. Otherwise, we can restart the process with higher value of $k$. Let $s_{\ell}$, $s_r$, $s_u$, and $s_b$ be the four segments containing the four extreme points, and denote as $\RR_{u\ell}$, $\RR_{ur}$, $\RR_{b\ell}$, and $\RR_{br}$ the four subsets of $\RR$ of segments intersected by the upper-left, upper-right, bottom-left, and bottom-right hulls respectively.

First, consider the case when $u$ (or $b$) lies on a vertical segment. Then it can be safely moved down (or up) until it hits an endpoint of its segment (i.e., a $0$th-order fixed point), or is no longer an extreme vertex. If it is no longer an extreme vertex, we restore the conditions of Lemma~\ref{lem:2orient-canonical-1} and continue with the new topmost (or bottommost) vertex of the transversal. Similarly, in the case when $\ell$ or $r$ lie on horizontal segments, we can slide them until they reach $0$th-order fixed points.

Assume then that $u$ and $b$ lie on horizontal segments, and $\ell$ and $r$ lie on vertical segments.
We further assume that the non-extreme vertices have been moved according to Lemma~\ref{lem:2orient-canonical-1}.
We observe that either (1) there exists a chain of the convex hull of $Q$ containing at least two endpoints of segments, (2) there exists a chain of the convex hull of $Q$ containing no endpoints, or (3) all four convex chains contain at most one endpoint.

In case (1), w.l.o.g., let the upper-left hull contain at least $2$ endpoints. Then we can slide $u$ left along its segment until it reaches its endpoint or an intersection with a line through two endpoints of segment in $\RR_{u\ell}$. Note that sliding $u$ left does not create problems on the top-right hull. We can also slide $\ell$ up along its segment until it reaches its endpoint or an intersection with a line through two endpoints of segments in $\RR_{u\ell}$. Thus, we also need to consider the intersections of segments in \RR with lines through pairs of endpoints; we will call these \emph{$1$st-order fixed points}. Now, vertices $u$ and $\ell$ are fixed, and we will proceed with  sliding $r$ and $b$.

For vertex $r$, we further distinguish two cases: (1.a) the upper-right convex hull contains at least two endpoints, or (1.b) it contains at most one endpoint.

In case (1.a), similarly to the case (1), we slide $r$ up until it reaches an endpoint of $s_r$ or an intersection with a line through two endpoints of segments in $\RR_{ur}$.

In case (1.b), we slide $r$ up, while unbending the strictly convex angle if it exists, until $r$ reaches an endpoint of $s_r$, or until the upper-right hull (which will be a straight-line segment $\overline{ur}$ at this point) contains a topmost or a rightmost endpoint of some segment in $\RR_{ur}$. In this case, $r$ will end up in an intersection point of a line passing through $u$ and an endpoint of a segment in $\RR_{ur}$. We will call such points \emph{$2$nd-order fixed points}, defined as an intersection of a segment in \RR with a line passing through a $1$st-order fixed point and an endpoint of a segment in \RR.

Similarly, for $b$ we distinguish two cases on the size of the lower-left hull, and slide $b$ until a fixed point (of $0$th-, $1$st-, or $2$nd-order). Thus, in case (1) we get a canonical convex transversal representation with strictly convex bends in the endpoints of segments in \RR, and extreme vertices in fixed points of $0$th-, $1$st-, or $2$nd-order. Note that there are $n^{2i+1}$ points that are $i$th-order, which means we only have to consider a polynomial number of discrete vertices in order to find a canonical solution.

In case (2), when there exists a chain of the convex hull of $Q$ that does not contain any endpoint, w.l.o.g., assume that this chain is the upper-left hull.
Then we slide $u$ left until the segment $\overline{u\ell}$ hits an endpoint of a segment while bending at an endpoint of the upper-right hull or at the point $r$, if the upper-right hull does not contain an endpoint. If the endpoint we hit does not belong to $s_u$, we are either in case (3), or there is another chain of the convex hull without endpoint and we repeat the process.
Else $u$ is an $0$th-order fixed point, and we can move $\ell$ up until it is an at most $1$st-order fixed point. Then, similarly to the previous case, we slide $r$ and $b$ until reaching at most a $1$st-order fixed point and at most a $2$nd-order fixed point respectively. Thus, in case (2) we also get a canonical convex transversal representation with strictly convex bends in the endpoints of segments in \RR, and extreme vertices in fixed points of $0$th-, $1$st-, or $2$nd-order.

Finally, in case (3), we may be able to slide some extreme points to reach another segment endpoint on one of the hulls, which gives us a canonical representation, as it puts us in case (1). However, this may not always be possible. Consider an example in Figure~\ref{fig:2-oriented-quad1}. If we try sliding $u$ left, we will need to rotate $\overline{u\ell}$ around the point $e_{u\ell}$ (see the figure), which will propagate to rotation of $\overline{\ell b}$ around $e_{b\ell}$, $\overline{br}$ around $e_{br}$, and $\overline{ru}$ around $e{ur}$, which may not result in a proper convex hull of $Q$.

The last case is a special case of a canonical partial convex transversal. It is defined by four segments on which lie the extreme points, and four endpoints of segments ``pinning'' the convex chains of the convex hull such that the extreme points cannot move freely without influencing the other extreme points. Next we present an algorithm to find an optimal canonical convex transversal with extreme points in the discrete set of fixed points. In the subsequent section we consider the special case when the extreme point are not necessarily from the set of fixed points.

\subsection{Calculating the canonical transversal}
    \begin{figure}
        \centering
        \begin{subfigure}{0.45\textwidth}
            \includegraphics[page=2,width=\textwidth]{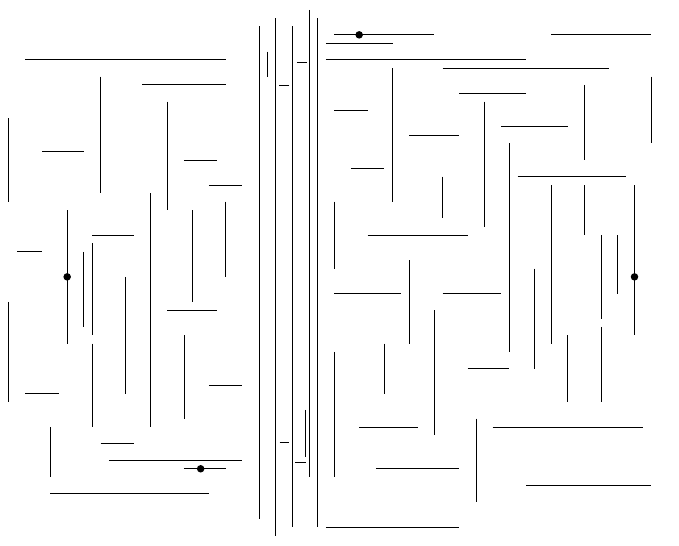}
            \caption{The input and our guess for the 4 extremal points.}
            \label{fig:2-oriented-limits}
        \end{subfigure}
        \hfil
        \begin{subfigure}{0.45\textwidth}
            \includegraphics[page=3,width=\textwidth]{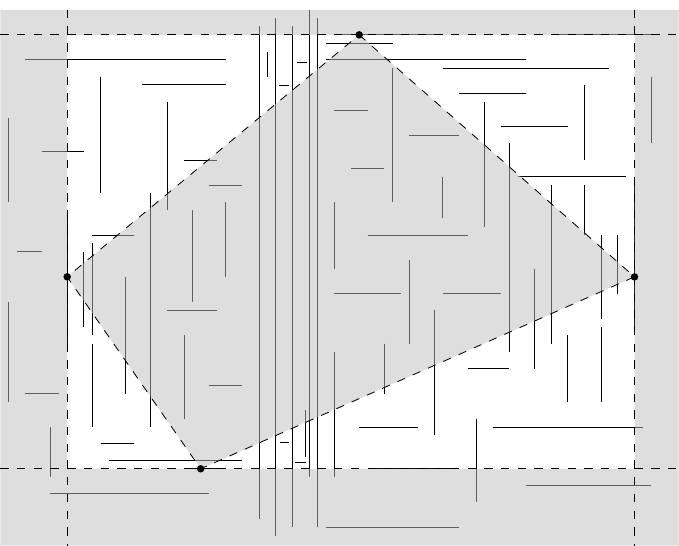}
            \caption{The regions induced by the extremal points.}
            \label{fig:2-oriented-regions}
        \end{subfigure}
        \begin{subfigure}{0.45\textwidth}
            \includegraphics[page=4,width=\textwidth]{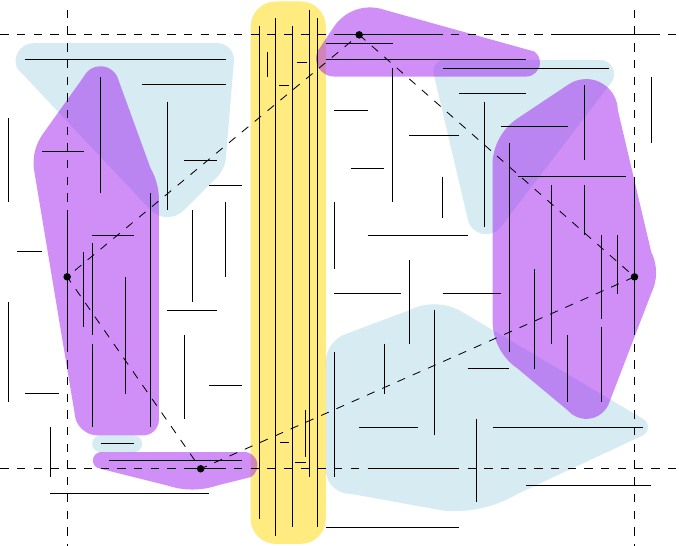}
            \caption{The different subproblems.}
            \label{fig:2-oriented-subproblems}
        \end{subfigure}
        \hfil
        \begin{subfigure}{0.45\textwidth}
            \includegraphics[page=5,width=\textwidth]{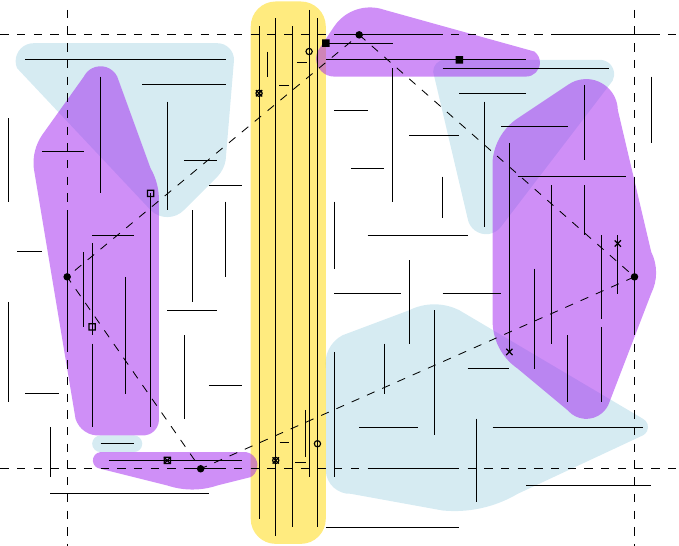}
            \caption{Our guesses for the endpoints of the subproblems.}
            \label{fig:2-oriented-guesses}
        \end{subfigure}
        \caption{The steps of our algorithm for 2-oriented segments.}
        \label{fig:2-oriented}
    \end{figure}

    \begin{figure}
        \centering
        \includegraphics{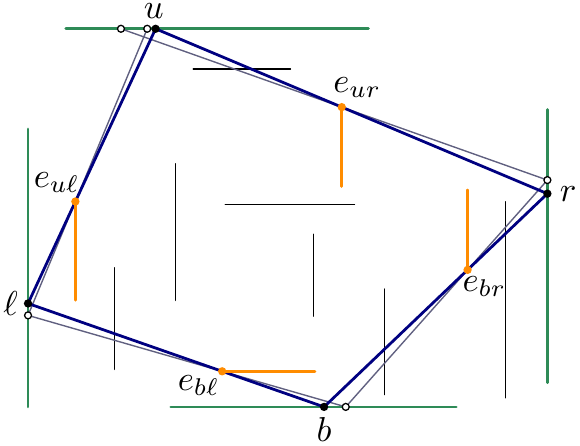}
        \caption{Extreme vertices are not in the set of fixed points.}
        \label{fig:2-oriented-quad1}
    \end{figure}

    We know that the vertices of a solution must lie on $0$th, $1$st or $2$nd order points. We subdivide the segments into several subproblems that we can solve similarly to the parallel case. First, we guess the four extreme points of our convex polygon, as seen in Figure~\ref{fig:2-oriented-limits}. These four points must be linked by \(x,y\)-monotone chains, and each chain has a triangular region in which its vertices may lie, as illustrated in Figure~\ref{fig:2-oriented-regions}. The key insight is that inside each of these regions, we have a partial ordering on the segments, as each segment can cross any \(x,y\)-monotone chain only once. This allows us to identify three types of subproblem:

    \begin{enumerate}
        \item Segments that lie inside two non-adjacent regions. We include any segments that lie between two of these into one subproblem and solve it separately. In Figure~\ref{fig:2-oriented-subproblems}, this subproblem is indicated in yellow. Note that there can be only one such subproblem.
        \item Segments that lie inside two adjacent regions. We include any segments that must come before the last one in our partial ordering. There are at most four of these subproblems; they are indicated in purple in Figure~\ref{fig:2-oriented-subproblems}.
        \item Segments that lie inside only one region. This includes only those segments not used in one of the other subproblems. There are at most four of these subproblems; they are indicated in blue in Figure~\ref{fig:2-oriented-subproblems}.
    \end{enumerate}

    For subproblems 1 and 2, we now guess the last points used on these subproblems (four points for subproblem 1, two for each instance of subproblem 2). We can then solve these subproblems using a similar dynamic programming algorithm to the one used for parallel segments. For subproblem 3, the problem is easier, as we are only building one half of the convex chain. We simply perform this algorithm for all possible locations for extreme points and endpoints of subproblems. For a given set of guesses, we can solve a subproblem of any type in polynomial time.

\subsection{Special case}\label{ch:special-subsection}
    As mentioned above this case only occurs when the four hulls each contain exactly one endpoint. The construction can be seen in Figure~\ref{fig:2-oriented-quad1}. Let $e_{u\ell}$, $e_{ur}$, $e_{br}$ and~$e_{b\ell}$ be the endpoints on the upper-left, upper-right, lower-right and lower-left hull. Let further $s_u$, $s_r$, $s_b$ and $s_\ell$ be the segments that contain the extreme points.

    For two points $a$ and $b$, let $l(a, b)$ be the line through $a$ and $b$.
    For a given position of $u$ we can place $r$ on or below the line $l(u, e_{ur})$. Then we can place $b$ on or left of the line $l(r, e_{br})$, $\ell$ on or above $l(b, e_{b\ell})$ and then test if $u$ is on or to the right of $l(\ell, e_{u\ell})$.
    Placing $r$ lower decreases the area where $b$ can be placed and the same holds for the other extreme points. It follows that we place $r$ on the intersection of $l(u, e_{ur})$ and $s_r$, we set $\{b\}=l(r, e_{br})\cap s_b$ and $\{\ell\}=l(b, e_{b\ell})\cap s_\ell$.
    Let then $u'$ be the intersection of the line $l(\ell, e_{\ell u})$ and the upper segment $s_u$.
    In order to make the test if $u'$ is left of $u$ we first need the following lemma.

    \begin{lemma}\label{th:math-does-not-increase}
        Given a line $\ell$, a point $A$, and a point $X(\tau)$ with coordinates $\left(\frac{P_1(\tau)}{Q(\tau)}, \frac{P_2(\tau)}{Q(\tau)}\right)$ where $P_1(\cdot)$, $P_2(\cdot)$, and $Q(\cdot)$ are linear functions. The intersection $Y$ of $\ell$ and the line through the points $X$ and $A$ has coordinates $\left(\frac{P'_1(\tau)}{Q'(\tau)}, \frac{P'_2(\tau)}{Q'(\tau)}\right)$ where $P'_1(\cdot)$, $P'_2(\cdot)$ and $Q'(\cdot)$ are linear functions.
    \end{lemma}
    \begin{proof}
        The proof consists of calculating the coordinates of the point $Y$ depending on $\tau$.

        Let $(a_x, a_y)$ be the coordinates of the point $A$.
        Let $l_1 x + l_2 y + l_3 =0$ and $k_1 x + k_2 y + k_3 =0$ be the equations of the lines $\ell$ and the line through $X$ and $A$.
        We can determine $k_2$ and $k_3$ depending on $k_1$ because the line passes through $X$ and $A$.
        It follows that $k_2=-k_1\frac{-P_1(\tau) + a_x Q(\tau)}{-P_2(\tau) + a_y Q(\tau)}$ and $k_3=-k_1\frac{a_y P_1(\tau) - a_x P_2(\tau)}{-P_2(\tau) + a_y Q(\tau)}$.
        We can then calculate the coordinates of the point $Y$. We obtain
        \begin{align*}
            Y= \left(\vphantom{\frac{P_1}{P_2}}\right. & \frac{-(a_y l_2 + l_3) P_1(\tau) + a_x (l_2 P_2(\tau) + l_3 Q(\tau))}{l_1 P_1(\tau) + l_2 P_2(\tau) - (a_x l_1 + a_y l_2) Q(\tau)},          \\
                             & \left.\frac{a_y l_1 P_1(\tau) - (a_x l_1 + l_3) P_2(\tau) + a_y l_3 Q(\tau)}{l_1 P_1(\tau) + l_2 P_2(\tau) - (a_x l_1 + a_y l_2) Q(\tau)}\right)\qedhere
        \end{align*}
      \end{proof}

    Let $(\tau, c)$ be the coordinates of the point $u$ for $\tau\in I$, where the constant $c$ and the interval $I$ are determined by the segment $s_u$.
    Then by Lemma~\ref{th:math-does-not-increase} we have that the points $r$, $b$, $\ell$, $u'$ all have coordinates of the form specified in the lemma.
    First we have to check for which values of $\tau$ the point $u$ is between $e_{u\ell}$ and $e_{ur}$, $r$ is between $e_{br}$ and $e_{ur}$, $b$ is between $e_{b\ell}$ and $e_{br}$ and $\ell$ is between $e_{b\ell}$ and $e_{u\ell}$. This results in a system of linear equations whose solution is an interval $I'$.

    We then determine the values of $\tau\in I'$ where $u'=\left(\frac{P_1(\tau)}{Q(\tau)}, \frac{P_2(\tau)}{Q(\tau)}\right)$ is left of $u=(\tau, c)$ by considering the following quadratic inequality: $\frac{P_1(\tau)}{Q(\tau)} \leq \tau$.
    If there exists a $\tau$ satisfying all these constraints, then there exists a convex transversal such that the points $u$, $r$, $b$ and $\ell$ are the top-, right-, bottom-, and leftmost points, and the points $e_{jk}$ ($j,k=u,r,b,\ell$) are the only endpoints contained in the hulls.

    Combining this special case with the algorithm in the previous section, we obtain the following result:

    \begin{theorem}\label{th:2-oriented-polynomial}
      Given a set of 2-oriented line segments, we can compute the maximum number of regions visited by a convex partial transversal in polynomial time.
    \end{theorem}

\subsection{Extensions}
    One should note that the concepts explained here generalize to more orientations. For each additional orientation there will be two more extreme points and therefore two more chains. It follows that for $\rho$ orientations there might be $\rho$th-order fixed points. This increases the running time, because we need to guess more points and the pool of discrete points to choose from is bigger, but for a fixed number of orientations it is still polynomial in $n$. The special case generalizes as well, which means that the same case distinction can be used.

\section{3-oriented intersecting segments}
\label{sec:hardthreeintersect}
We prove that the problem of finding a maximum convex partial transversal $Q$ of a set of 3-oriented segments $\RR$
is NP-hard using a reduction from Max-2-SAT.

\begin{theorem}
Let $\RR$ be a set of segments that have three different orientations. The problem of finding a maximum convex partial transversal $Q$ of $\RR$ is NP-hard.
\end{theorem}

First, note that we can choose the three orientations without loss of generality: any (non-degenerate) set of three orientations can be mapped to any other set using an affine transformation, which preserves convexity of transversals.
We choose the three orientations in our construction to be vertical ($|$), the slope of $1$ ($\+$) and the slope of $-1$ ($\-$).

Given an instance of \textsc{Max-2-SAT} we construct a set of segments $\RR$ and then we prove that from a maximum convex partial transversal $Q$ of $\RR$ one can deduce the maximum number of clauses that can be made true in the instance.

\subsection{Overview of the construction}
Our constructed set \(\RR\) consists of several different substructures.
The construction is built inside a long and thin rectangle, referred to as the \emph {crate}.
The crate is not explicitly part of $\RR$.
Inside the crate, for each variable, there are several sets of segments that form chains. These chains alternate $\+$ and $\-$ segments reflecting on the boundary of the crate.
For each clause, there are vertical $|$ segments to transfer the state of a variable to the opposite side of the crate.
Figure~\ref {fig:bananaoverview} shows this idea.
However, the segments do not extend all the way to the boundary of the crate; instead they end on the boundary of a slightly smaller convex shape inside the crate, which we refer to in the following as the \emph{banana}.
Figure~\ref {fig:bananabending} shows such a banana.
Aside from the chains associated with variables, \(\RR\) also contains segments that form gadgets to ensure that the variable chains have a consistent state, and gadgets to represent the clauses of our \textsc{Max-2-SAT} instance.
Due to their winged shape, we refer to these gadgets by the name \emph{fruit flies}. (See Figure~\ref{fig:fruitfly} for an image of a fruit fly.)

Our construction makes it so that we can always find a transversal that includes all of the chains, the maximum amount of segments on the gadgets, and half of the $|$ segments. For each clause of our \textsc{Max-2-SAT} instance that can be satisfied, we can also include one of the remaining $|$ segments.

\begin{figure}
\hspace{-1cm}\includegraphics{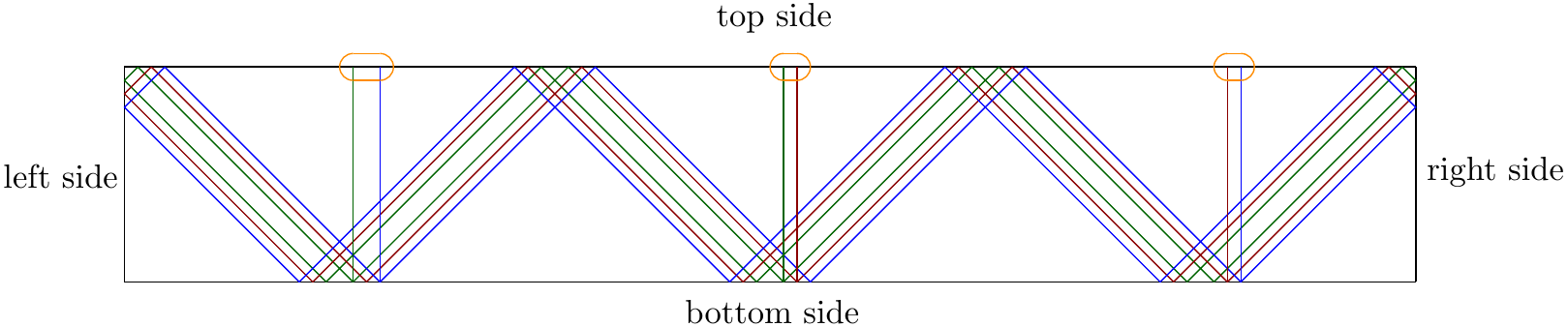}
\caption{Overview of our construction. Each of the colored segment chains represents a variable. At each point where a chain bounces on the banana there is a fruit fly gadget. At each area marked orange there is a clause gadget. Each chain is only pictured once, but in actuality each chain is copied \(m+1\) times and placed at distance \(\epsilon\) of each other. The distance between the different variables is exaggerated for clarity.}
\label{fig:bananaoverview}
\end{figure}

\subsection{Complete construction}
In the following we assume that we are given an instance of \textsc{Max-2-SAT} $(V, C)$, where $V=\{v_1, \dots, v_n\}$ is the set of variables and $C=\{c_1, \ldots, c_m\}$ is the set of clauses. For an instance of \textsc{Max-2-SAT}, each clause has exactly two literals. The goal is to find an assignment for the variables such that the maximum number of clauses is satisfied.
We first construct a set of segments $\RR$ and then we prove that from a maximum convex partial transversal $Q$ of $\RR$ one can deduce the maximum number of clauses that can be made true in $(V,C)$.

The different substructures of \(\RR\) have sizes of differing orders of magnitude. Let therefore $\alpha$ (distance between the chains of the different variables), $\beta$ (distance between the inner and outer rectangles that shape the banana), $\gamma$ (horizontal distance between the inner anchor points of a fly), $\delta$ (vertical distance between outermost bristle line and a fly's wing), $\epsilon$ (distance between multiple copies of one variable segment), and \(\zeta\) (length of upper wing points of the flies), be rational numbers (depending polynomially on $n$ and $m$) with $1\gg\alpha\gg\beta\gg\gamma\gg\delta\gg\epsilon\gg\zeta>0$. Usable values for these constants are given in Table~\ref{tab:bananaconstants}, but other values are also possible.

\begin{table}
\centering

\begin{tabular}{cc}
\toprule
Constant & Value \\
\midrule
\(\alpha\) & \(\frac{1}{100n}\) \\
\(\beta\) & \(\frac{\alpha}{100}\) \\
\(\gamma\) & \(\frac{\alpha^3}{10000m^2}\) \\
\(\delta\) & \(\frac{\gamma}{100m}\) \\
\(\epsilon\) & \(\frac{\delta}{100}\) \\
\(\zeta\) & \(\frac{\delta}{100m^2}\) \\
\bottomrule
\vspace{0cm}
\end{tabular}
\caption{Possible values for the constants used in our hardness construction.}
\label{tab:bananaconstants}
\end{table}
\subsubsection{Construction of the Chains}
\label{sec:constructchains}

First we create the crate \(B\) with sides \(1\) by \(2m\). Then we construct the chains.

\begin{lemma}
For each variable, we can create a closed chain of \(4m+2\) segments with endpoints on \(B\) by alternating \(\-\) and \(\+\) segments.
\end{lemma}
\begin{proof}
 Let $v_i$ be a variable and \(s_i\) be the (closed) chain we construct to be associated with \(v_i\).
Then the first segment of $s_i$ starts close to the top left of \(B\) at coordinates \((i\alpha,1)\) and has orientation $\-$ until it hits \(B\) at coordinates \((i\alpha+1,0)\) so that it connects the top and bottom sides of \(B\). Then the chain reflects off  \(B\)'s bottom side so the second segment has orientation $\+$, shares an endpoint with the first segment and again connects \(B\)'s bottom and top sides by hitting \(B\) at point \((i\alpha+2, 1)\). Then we reflect downwards again.

Every time we go downwards and then upwards again we move a distance of \(2\) horizontally. Since our rectangle has length \(2m\) we can have \(m-1\) pairs of segments like this. We are then at the point \((i\alpha+2m -2,1)\). If we then go downwards to \((i\alpha+2m-1,0)\) and reflect back up again, we hit the right side of \(B\) at coordinates \((2m,1-i\alpha)\). We reflect back to the left and hit the top side of \(B\) at \(2m-i\alpha, 1\). This point is symmetrical with our starting point so as we reflect back to the left we will eventually reach the starting point again, creating a closed chain, no matter our values of \(i\) and \(\alpha\).
\end{proof}

We construct a chain \(s_i\) for each variable \(v_i \in V\). Then we construct two \(|\) segments for each clause \(c_j \in C\) as follows:\\ Let \(v_k, v_l\) be the variables that are part of clause \(c_j\). There are shared endpoints for two segments of \(s_k\) and \(s_l\) at \((k\alpha+1+2(j-1),0)\) and \((l\alpha+1+2(j-1),0)\) respectively. At these points, we add \(|\) segments, called \emph{clause segments}, with their other endpoints on the top side of \(B\). (See Figure~\ref{fig:bananaoverview}.)

Each chain is replaced by \(m+1\) copies of itself that are placed at horizontal distance \(\epsilon\) of each other. The clause segments are not part of the chain proper and thus not copied.
\subsubsection{How to Bend a Banana}
\label{sec:bananabending}

\begin{figure}
\includegraphics{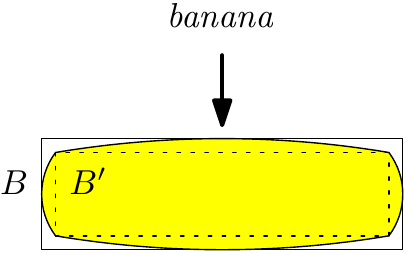}
\centering
\caption{A banana in a crate. It consists of four parabolas with very low curvature so that it is almost as straight as a rectangle while being strictly convex. The distance between \(B\) and \(B'\) has been exaggerated in the figure for clarity.}
\label{fig:bananabending}
\end{figure}

In the previous section we constructed the variable chains. In Section \ref{sec:constructfruitflies} we will construct fruit fly gadgets for each reflection of each chain and each clause. However, if we place the fruit flies on \(B\) they will not be in strictly convex position. To make it strictly convex we create a new, strictly convex bounding shape (the banana) which we place inside the crate. The fruit flies are then placed on the boundary of the banana, so that they are in strictly convex position.

We are given our crate \(B\) of size $1$ by $2m$ and we wish to replace it with our banana shape. (See Figure~\ref {fig:bananabending}.)
We want the banana to have certain properties.
\begin {itemize}
  \item The banana should be strictly convex.
  \item The distance between the banana and the crate should be at most $\beta$ everywhere.
  \item We want to have sufficiently many points with rational coordinates to lie on the boundary of the banana.
\end {itemize}

To build such a banana, we first create an {\em inner crate} $B'$, which has distance $\beta$ to $B$. Then we create four parabolic arcs through the corners of $B'$ and the midpoints of the edges of $B$; see Figure~\ref {fig:bananabending}.
In the following, we make this construction precise. We start by focusing on the top side of the banana.

Let \(P\) be the parabola on the top side. It goes from \((\beta,1- \beta)\), through vertex \((m,1)\) to point \((2m-\beta,1-\beta)\). This means the equation defining \(P\) is \(y = \frac{-\beta}{(m-\beta)^2}(m-x)^2 + 1\).
For the top side of the banana, we have that there are two types of flies that need to be placed. At each reflection of a chain there is a \emph{reflection fly}. For each clause there is a \emph{clause fly}. The positions of the reflections on the top side for a variable \(v_i\)'s chains are \((i\alpha+2k + j\epsilon,1)\) and \((2m - [i\alpha+2k + j\epsilon],1)\) for each \(j \in \{1,\dots,m+1\}, k \in \{0,\dots,m-1\}\) (see Section \ref{sec:constructchains}). The clauses have approximate position \((1 +j, 1)\) for each \(j\in \{1,\dots,m\}\). (The clauses are moved horizontally with a factor of \(\alpha\) depending on which variables are included in the clause.) The reflection flies have distance \(\alpha\) from each other, which is much more than \(\beta\) and \(\epsilon\). The distances between reflections on the other sides of the crate are of similar size. So it is possible to create a box with edges of size \(\beta\) around each reflection point such that
\begin{itemize}
\item All involved segments in the reflection (meaning every copy of the two chain segments meeting at the reflection, and a possible clause segment) intersect the box.
\item There cannot be any other segments or parts of flies inside of the box
\item The top and bottom edges of the box lie on \(B\) and \(B'\). (Or the left and right edges do if the reflection is on the left or right side of the crate.)
\end{itemize}
The clause flies have distance 1 from any other flies, so even though they are wider than reflection flies (width at most \(n\alpha\)), we can likewise make a rectangular box such that both the clause's segments intersect the box and no other segments do. The box has a height of \(\beta\) and is placed between \(B\) and \(B'\); the width is based on what is needed by the clause.

Inside each reflection- or clause box we find five points on \(P\) with rational coordinates which will be the fly's \emph{anchor points}. We take the two intersection points of the box with \(P\), the point on \(P\) with \(x\)-coordinate center to the box, and points \(\gamma\) to the left and right of this center point on \(P\). We will use these anchor points to build our fly gadgets. This way we both have guaranteed rational coordinates and everything in convex position. See Figure \ref{fig:flybox}.

\begin{figure}
\centering
\includegraphics[width=0.6\textwidth,trim={6.7cm, 0, 7cm, 0},clip]{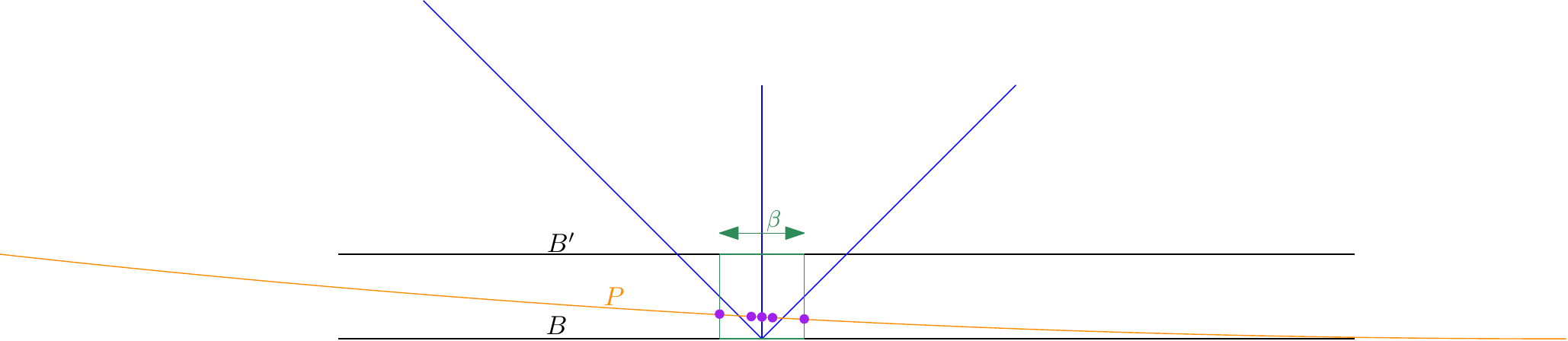}
\caption{Close-up showing the box we construct around a reflection of a chain on the bottom side of \(B\). \(P\) is the parabola that forms the bottom side of the banana. The anchor points (shown in purple) will be used to construct fruit flies. Going from left to right, the second and fourth points have a horizontal distance \(\gamma\) to the center point. In the figure \(\gamma\) is exaggerated for clarity.}
\label{fig:flybox}
\end{figure}

\subsubsection{Construction of the Fruit Flies}
\label{sec:constructfruitflies}

\begin{figure}
        \centering
        \begin{subfigure}{0.45\textwidth}
            \includegraphics[width=\textwidth]{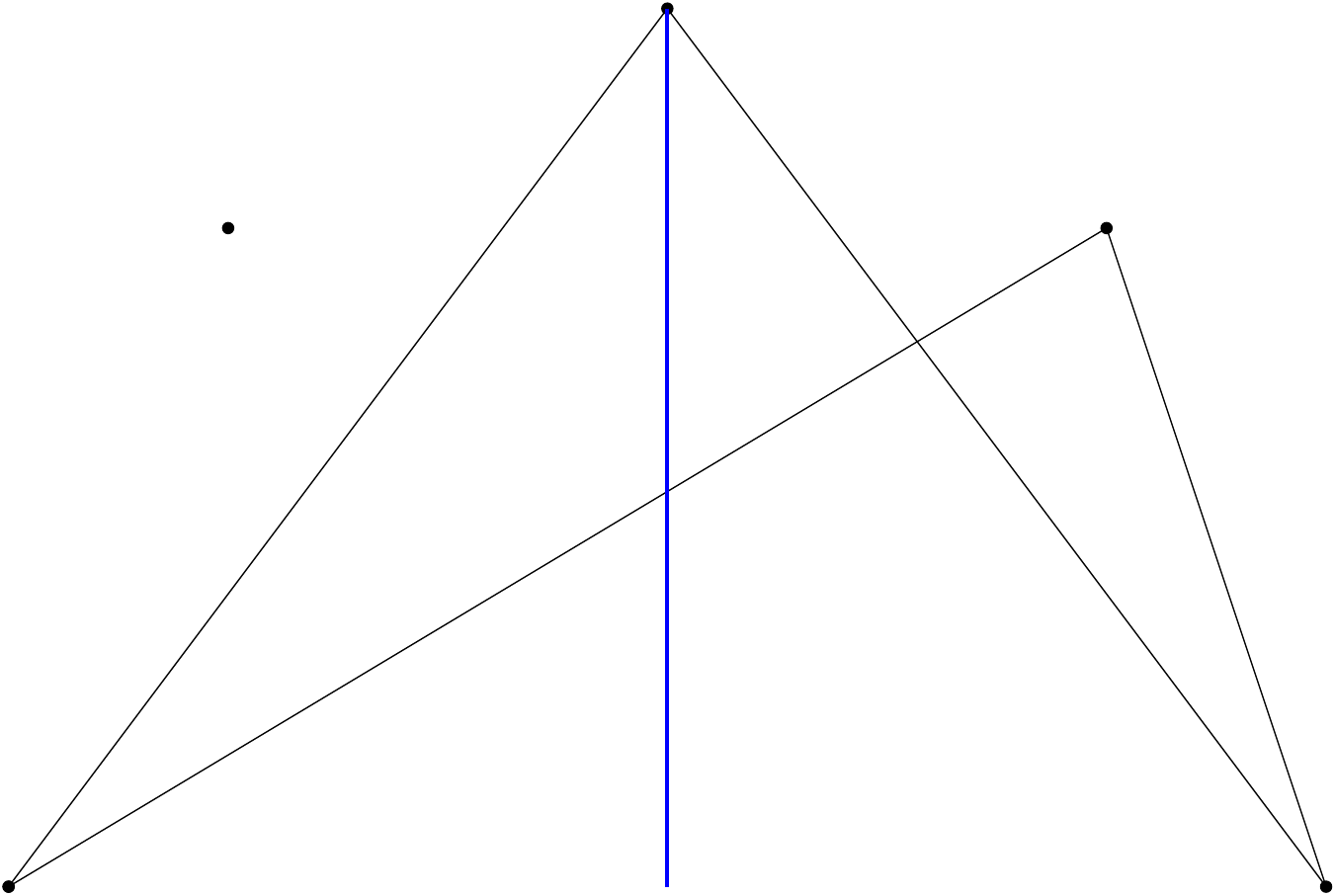}
			\caption{Fly with clause segment on left wing}
            \label{fig:leftwingfly}
        \end{subfigure}
        \begin{subfigure}{0.45\textwidth}
            \includegraphics[width=\textwidth]{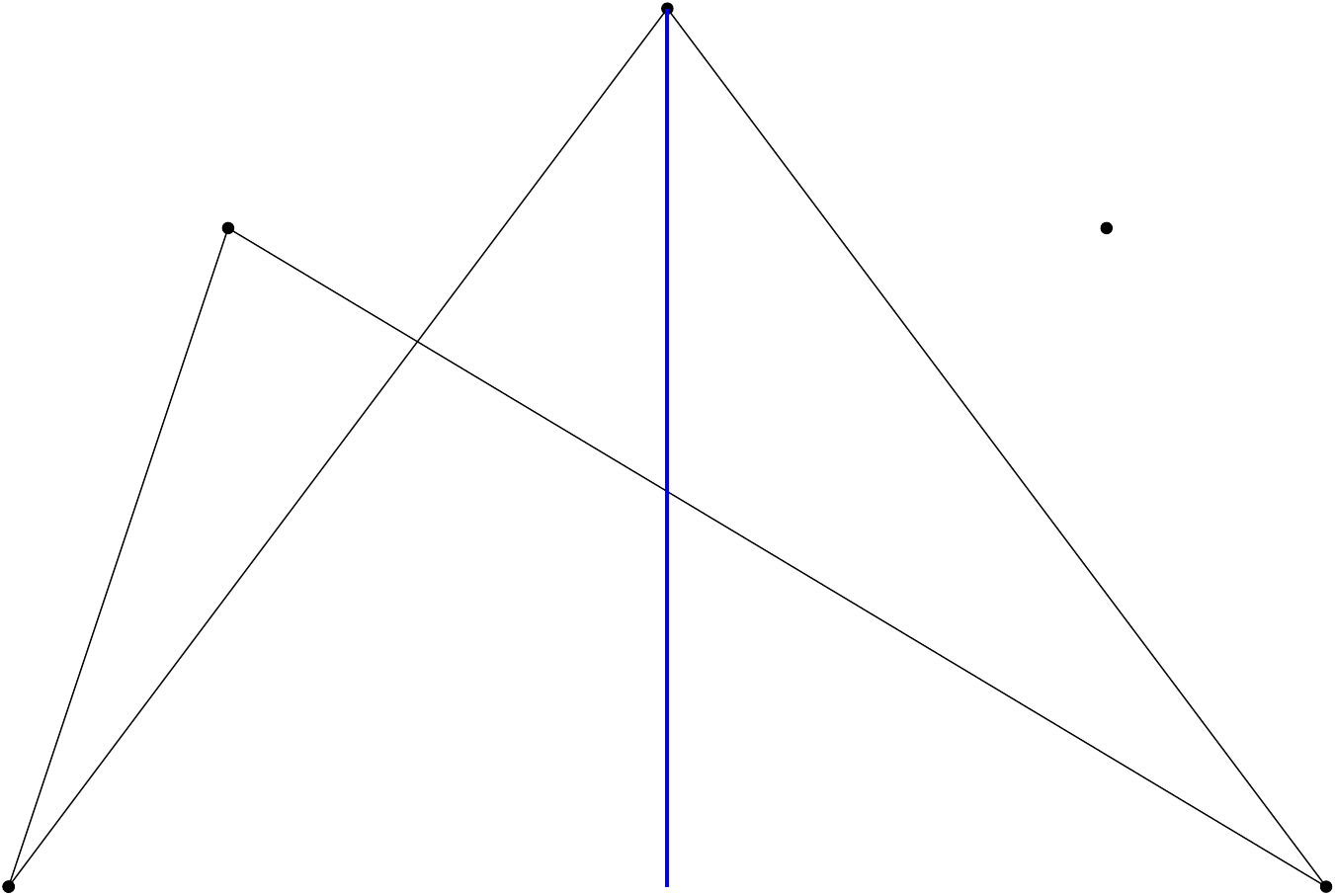}
            \caption{Fly with clause segment on right wing}
            \label{fig:rightwingfly}
        \end{subfigure}
        \caption{By choosing which anchor points to connect, we can make the clause segment intersect either the left or the right wing of the fly. The clause segment always intersects the center anchor point (before it is shortened to be in convex position). }
        \label{fig:politicalfly}
    \end{figure}

 \begin{figure}
        \centering
            \includegraphics[width=0.7\textwidth]{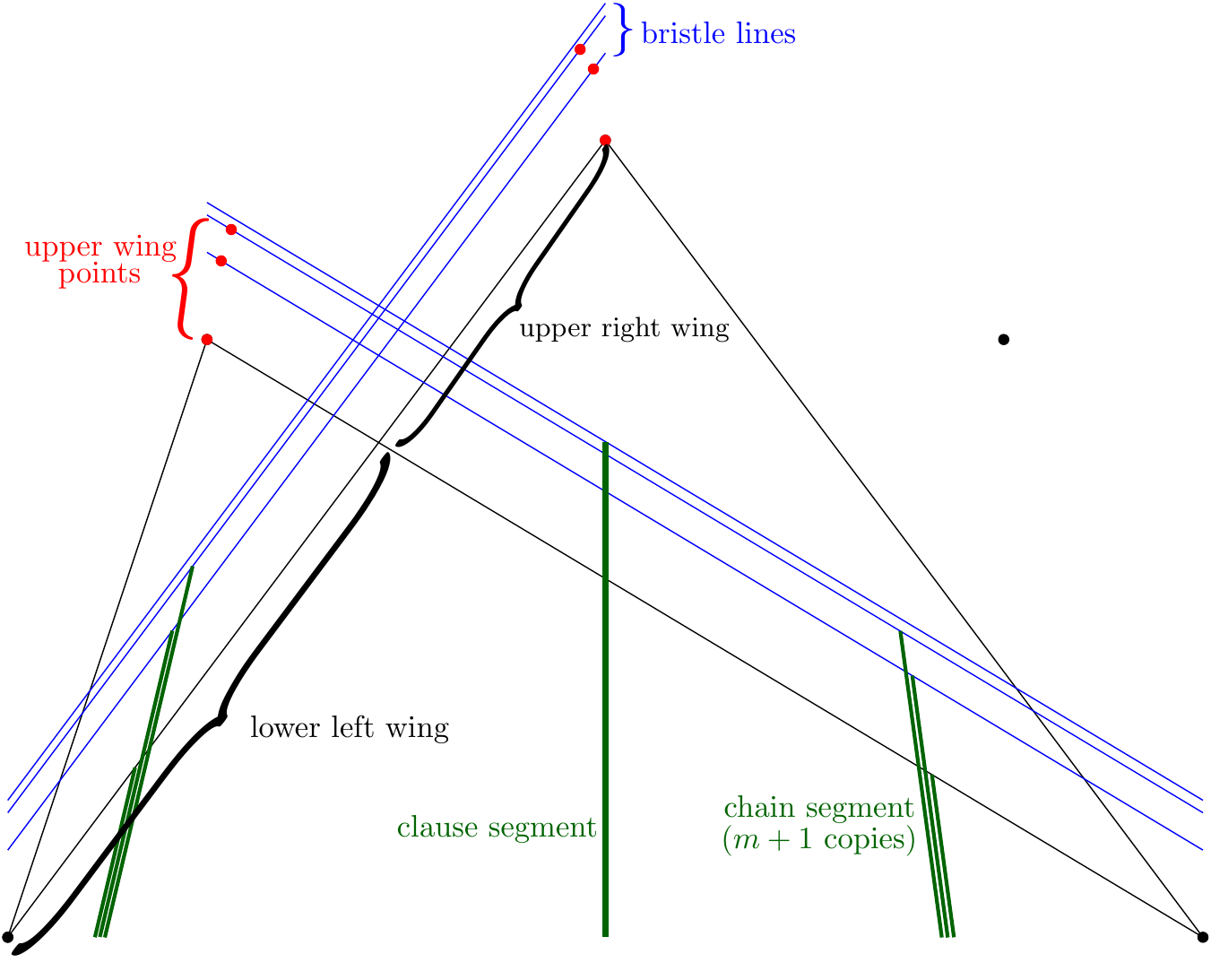}
        \caption{The Fruit Fly gadget. The endpoints of the segments and the upper wing points (shown in red) are in convex position due to their placement on the bristle lines (shown in blue). In our actual construction, the fly appears completely \emph{swatted}: the points defining the fly lie on the same parabola, so close together that the fly is almost completely flat. The chain segments are then at an angle of \(90^\circ\).}
        \label{fig:fruitfly}
    \end{figure}

Now we construct the numerous fruit fly gadgets that ensure the functionality of our construction.
The general concept of a fruit fly can be seen in Figure~\ref{fig:fruitfly}.

NB: In our images all flies are oriented "top side up", as if they are on the top side of the banana. (Even though reflection flies with a clause segment as seen in Figure~\ref{fig:fruitfly} can occur on the bottom side). In our actual construction, they are oriented such that the "top side" is pointing outwards from the banana relative to the side they are on.

We create
$n(4m+2)+m$ flies in total, of two types:
one {\em reflection fly} at each reflection of each chain and
one {\em clause fly} for each clause.


Each fly consists of a pair of {\em wings}. The wings are created by connecting four of the five anchor points in a criss-cross manner, creating two triangles. The intersection point between the segments in the center of the fly divides the wings into an \emph{upper wing} and a \emph{lower wing}. The intersecting segments are referred to in the following as the \emph{wing lines}.

The choice of which anchor points to connect depends on the presence of a clause segment. We always connect the outer points and the center point. If the fly is a reflection fly with a clause segment, we connect the anchor point such that the segment intersects the lower right wing if the variable appears negated in the clause and the lower left wing otherwise.  See Figure \ref{fig:politicalfly}. If the fly is a clause fly, or a reflection fly without a clause segment, the choice is made at random. The wings are implicit, there are no segments in \(\RR\) that correspond to them.

Besides the segments making up the wings, we also create \(m+1\) line segments parallel to the wing lines at heights increasing up to \(\delta\) above each wing line. We will refer to these extra line segments as the flies' \emph{bristle lines}.
The distance between bristle lines decreases quadratically. We first compute a step size \(\kappa = \frac{\delta}{m^2}\). Then, for both wing lines, the bristle lines are placed at heights \(h_w + \delta - (m+1-i)^2\kappa\) for all \(i \in \{1,\dots,m+1\}\) where \(h_w\) is the height of the wing line (relative to the orientation of the fly).

When the bristle lines have been constructed, we shorten all of the line segments involved in the fly so that their endpoints are no longer on \(B\) but lie on a wing line or one of the bristle lines. We do this in such a manner that the endpoints of the copies of a segment are all in convex position with each other (as well as with the next and previous fly) and have rational coordinates.

To shorten the chain segments, we consider the copies as being sorted horizontally. W.l.o.g. we look at the segments intersecting the lower left wing, which are sorted from left to right. The first segment (the original uncopied one) gets shortened so its endpoint is the intersection between the original line segment and the fly's wing line. The next segment gets shortened so its endpoint is the intersection with the first bristle line. The next segment gets shortened so its endpoint is on the second bristle line, etc. The final copy's endpoint will lie on the penultimate bristle line.
If there is a clause segment on the wing it is shortened so that its endpoint lies on the highest bristle line.

After shortening all of the segments we also add \(m+1\) vertical line segments of length \(\zeta\) to each of the flies' upper wings. Their length is chosen to be this short so they behave like points and it does not matter for any transversal which point on the segment is chosen, but only if the segment is chosen at all. The line segments have horizontal distance \(\epsilon\) from each other and are placed in a reverse manner compared to the segments; so the first line segment intersects penultimate bristle line, the next segment intersects the bristle line below that, etc. until the final segment intersects the wing line at the tip of the wing. These segments form a part of \(\RR\).  These segments (shown in red in Figure~\ref{fig:fruitfly}) are referred to in the following as the fly's \emph{upper wing points}.

For clause flies, the clause segments are shortened such that their endpoints lie on the highest bristle line. The fly is placed such that each clause segment has its own wing. The clause flies also have \(m+1\) upper wing points per wing.

\subsubsection{Putting it all together}
\label{sec:bananacount}
\begin{lemma}
\label{lem:polybanana}
The transformation of the instance of \textsc{Max-2-SAT} to \textsc{3-Oriented Maximum Partial Transversal} can be done in polynomial time and space.
\end{lemma}
\begin{proof}
The set of all segments of all copies of the chains, the clause-segments used for the clauses and the points of the flies together form the set $\RR$. Each chain has \(4m+2\) segments. We have \(n(m+1)\) chains. We also have \(n(4m+2)\) reflection flies that each have \(2(m+1)\) upper wing points. We have \(2m\) clause segments. Finally we have \(m\) clause flies that include \(2(m+1)\) upper wing points each. That brings the total size of \(\RR\) to
\(n(12m^2 + 18m + 6)+ 4m + 2m^2\).
During construction, we create \(5\) anchor points for each fly. We also construct \(m+1\) bristle lines for each fly. Each segment in \(\RR\) is first created and then has its endpoints moved to convex position by intersecting the segment with a bristle line. All of this can be done in polynomial time and space.
\end{proof}
\subsection{Proof of Correctness}

\begin{lemma}
The Max-2-SAT instance $(V, C)$ has an assignment of variables such that $k$ clauses are true if and only if the set $\RR$ allows a maximum convex partial transversal $Q$, with $|Q|=|\RR|-n(4m+2)(m+1) - (m-k)$.
\end{lemma}
\begin{proof}

Our fruit fly gadget is constructed such that a convex transversal can only ever include half of the upper wing points for each fly. So any transversal will at most include \(|\RR| - n(4m+2)(m+1) \) points.
As we will show below, it is always possible to create a transversal that includes half of the upper wing points of every fly, plus every chain segment of every variable, giving a transversal of at least \(|\RR|-n(4m+2)(m+1) - m\) points. So, as each fly has \(m+1\) upper wing points and each segment is copied \(m+1\) times, a maximum traversal must visit all flies and all chain segments, no matter how many clause segments are included.
A guaranteed way to include all flies in a transversal is to stay on the edge of the banana and including the flies in the order they are on the banana's boundary, while only choosing points on the chain segments that are inside of the halfplanes induced by the flies' wing lines. (For convenience, we assume that we only pick the segments' endpoints as it doesn't matter which point we choose in this halfplane in regard to which other points are reachable while maintaining convexity and the ability to reach other flies.)\\
\\
The only way to visit the maximum number of regions on a fly is to pick one of the two wing lines (with related bristle lines) and only include the points on the upper and lower wing it induces. (So either all segments on the lower left and upper right wing, or all segments on the lower right and upper left wing.) If we consider the two flies that contain the opposite endpoints of a chain segment it is clear that choosing a wing line on one of the flies also determines our choice on the other fly. If we choose the wing line on the first fly that does not include the \(m+1\) copies of the line segment we must choose the wing line on the other fly that does include them, otherwise we miss out on those \(m+1\) regions in our transversal. Since the chains form a cycle, for each set of chains corresponding to a variable we get to make only one choice. Either we choose the left endpoint of the first segment, or the right one. The segments of the chain then alternate in what endpoints are included.
If we choose the left endpoint for the first segment of a chain it is equivalent to setting the corresponding variable to true in our \textsc{Max-2-SAT} instance, otherwise we are setting it to false. \\ \\
At the reflection flies that have a clause segment, the endpoint of that clause segment on the fly can be added to the partial transversal iff it is on the wing that is chosen for that fly.(Recall from Section \ref{sec:constructfruitflies} that which wing contains the clause segment's endpoint depends on if the variable appears negated in the clause.) The clause segment has a clause fly at its other endpoint which it shares with the clause segment of another variable. If one of the two clause segments is already included in the transversal because it was on the correct wing of the reflection fly, we can choose the other wing of the clause fly. If neither of the clause segments are already included, we can only include one of the two by picking one of the two wings. This means there is no way to include the other clause segment, meaning our convex partial transversal is \(1\) smaller than it would be otherwise. This corresponds to the clause not being satisfied in the 2-SAT assignment.\\
Since we can always get half of the upper wing points and all of the chain segments our maximum convex partial transversal has cardinality $|Q|=|\RR|-n(4m+2)(m+1) - (m-k)$, where \(k\) is the number of clauses that can be satisfied at the same time. Since Lemma \ref{lem:polybanana} shows our construction is polynomial, we have proven that the problem of finding a maximum convex transversal of a set of line segments with 3 orientations is NP-hard.
\end{proof}

\subsection {Implications}

Our construction strengthens the proof by~\cite{schlipf2012notes} by showing that using only 3 orientations, the problem is already NP-hard. The machinery appears to be powerful: with a slight adaptation, we can also show that the problem is NP-hard for axis-aligned rectangles.

\begin{theorem}
Let $\RR$ be a set of (potentially intersecting) axis-aligned rectangles. The problem of finding a maximum convex partial transversal $Q$ of $\RR$ is NP-hard.
\end{theorem}

\begin{proof}
  We build exactly the same construction, but afterwards we replace every vertical segment by a $45^\circ$ rotated square and all other segments by arbitrarily thin rectangles. The points on the banana's boundary are opposite corners of the square, and the body of the square lies in the interior of the banana so placing points there is not helpful.
\end{proof}

\bibliographystyle{plain}
\bibliography{bibliography}

\end{document}